\documentclass[12pt]{amsart}
\usepackage{amsfonts}
\usepackage{epsfig}
\usepackage{graphicx}
\usepackage{amsmath}
\usepackage{amssymb}
\usepackage{latexsym}
\usepackage{rotating}
\usepackage{pstricks, pst-node, pst-text, pst-3d}
\usepackage{amsbsy}

\usepackage{amsthm}
\usepackage{mathrsfs}
\usepackage[all]{xy}
\usepackage{a4wide}

\input{amssym.def}
\newtheorem{theorem}{Theorem}
\newtheorem{proposition}{Proposition}

\newtheorem{lemma}{Lemma}
\newtheorem{remark}{Remark}

\newtheorem{definition}{Definition}
\theoremstyle{remark}

\newcommand{\re}{\text{\rm Re }}

\newcommand{\tb}{\text{\bf t}}
\newcommand{\Diff}{\text{\rm Diff}\,}

\newcommand{\Vect}{\text{\rm Vect$\,$}}

\newcommand{\Hol}{\text{\rm Hol }}

\DeclareMathOperator{\spn}{span}

\DeclareMathOperator{\Gr}{Gr}

\DeclareMathOperator{\kernel}{ker}

\DeclareMathOperator{\cokernel}{coker}
\DeclareMathOperator{\virtcard}{virtcard}
\DeclareMathOperator{\virtdim}{virtdim}
\DeclareMathOperator{\ind}{ind}
\DeclareMathOperator{\id}{id}
\begin{document}
\title[Evolution of smooth shapes]{Evolution of smooth shapes and integrable systems}
\author{Irina Markina and Alexander Vasil'ev}

\address{Department of Mathematics,
University of Bergen, P.O.~Box~7803, Bergen N-5020, Norway}

\email{irina.markina@uib.no} \email{alexander.vasiliev@uib.no}

\thanks{ The authors have been  supported by the grant of the Norwegian Research Council \#204726/V30, by the NordForsk network `Analysis and Applications', grant \#080151, and by the European Science Foundation Research Networking Programme HCAA}

%\thanks{}



%\date{04/08 2011}

\begin{abstract}
We consider a homotopic evolution in the space of smooth shapes starting from the unit circle. Based on the L\"owner-Kufarev equation we give a Hamiltonian formulation of this evolution and provide conservation laws. The symmetries of the evolution are given by the Virasoro algebra.
The `positive' Virasoro generators span the holomorphic part of the complexified vector bundle over the space of conformal embeddings of the unit disk into the complex plane and smooth on the boundary. In the covariant formulation they are conserved along the Hamiltonian flow.
The `negative' Virasoro generators can be recovered by an iterative method making use of the canonical Poisson structure. We  study an embedding of the L\"owner-Kufarev trajectories into the Segal-Wilson  Grassmannian, construct the $\tau$-function, the Baker-Akhiezer function, and finally, give a class of solutions to the KP equation.
\end{abstract}
\maketitle

\section{Introduction}\label{introd}

The challenge of structural understanding of non-equilibrium interface dynamics has become increasingly important in mathematics and physics. 
The study of 2D shapes is one of the central problems in the field of applied sciences. A program of such study and its importance was summarized by Mumford at ICM 2002 in Beijing \cite{Mumford}. By {\it shape} we understand a  simple closed curve in the complex plane dividing it into two simply connected domains.
Dynamical interfacial properties, such as fluctuations, nucleation and aggregation, mass and charge transport, are often very complex. There exists no single theory or model that can predict all such properties. Many physical processes, as well as complex dynamical systems, iterations and construction of Lie semigroups with respect to the composition operation, lead to the study of growing systems of plane domains. Recently, it has become clear that one-parameter expanding evolution families of simply
connected domains in the complex plane in some special models has been governed by infinite systems
of evolution parameters, conservation laws. This phenomenon reveals a bridge between a non-linear evolution of complex shapes emerged in physical problems, dissipative in most of the cases, and exactly solvable models. One of such processes is the Laplacian growth, in which the harmonic (Richardson's) moments are conserved under the evolution, see e.g., \cite{GustVas, Mineev}. The infinite number of evolution parameters reflects the infinite number of degrees of freedom of the system, and clearly suggests to apply field theory methods as a natural tool of study. The Virasoro algebra provides a structural background in most of field theories, and it is not surprising that it appears in soliton-like problems, e.g., KP, KdV or Toda hierarchies, see \cite{Faddeev, Gervais}.

Another group of models, in which the evolution is governed by an infinite number of parameters, can be observed in controllable dynamical systems, where the infinite number of degrees of freedom follows from the infinite number of driving terms. Surprisingly,
the same algebraic structural background appears again for this group. We develop this viewpoint in the present paper.

One of the general approaches to the homotopic evolution of shapes starting from a canonical shape, the unit disk in our case,  was provided by L\"owner and Kufarev  \cite{Kufarev, Loewner, Pommerenke2}. A shape evolution is described by a time-dependent conformal parametric map from the canonical domain onto the domain
bounded by the shape at any fixed instant. In fact, these one-parameter conformal maps satisfy the L\"owner-Kufarev differential equation, or an infinite dimensional controllable system, for which
the infinite number of conservation laws is given by the {\it Virasoro generators} in their covariant form.

Recently, Friedrich and
Werner \cite{FriedrichWerner}, and independently Bauer and Bernard \cite{BB}, found relations between SLE (stochastic or Schramm-L\"owner evolution) and the highest weight representation of the Virasoro algebra. Moreover, Friedrich developed the Grassmannian approach to relate SLE with the highest weight representation of the Virasoro algebra in~\cite{Friedrich}.

All  above results encourage us to conclude that the {\it Virasoro algebra} is a common algebraic structural basis for these and possibly other types of contour dynamics and we present the development in this direction here. At the same time, the infinite number of conservation laws suggests
a relation with exactly solvable models.

The geometry underlying classical integrable systems is reflected in Sato's  and Segal-Wilson's constructions of the infinite dimensional {\it Grassmannian} $\Gr$. Based on the idea that the evolution of shapes in the plane is related to an evolution in a general universal space, the Segal-Wilson Grassmannian in our case, we provide
an embedding of the L\"owner-Kufarev evolution into a fiber bundle with the cotangent bundle over $\mathcal F_0$ as a base space, and with the smooth Grassmannian $\Gr_{\infty}$ as a typical fiber. Here $\mathcal F_0$ denotes the space of all conformal embeddings $f$ of the unit disk into $\mathbb C$ normalized by $f(z)=z\left(1+\sum_{n=1}^{\infty}c_nz^n\right)$ smooth on the boundary $S^1$, and under the  {\it smooth Grassmannian} $\Gr_{\infty}$ we understand  a dense subspace $\Gr_{\infty}$ of $\Gr$ defined further in Section 4.

 We intent to keep the paper self-sufficient, and its structure  is as follows. Section 2 provides the reader with necessary definitions and structures of the Virasoro-Bott group, of the group $\Diff S^1$ of orientation preserving diffeomorphisms of the unit circle $S^1$, and of Kirillov's homogeneous manifold $\Diff S^1/ S^1$, as well as of their infinitesimal descriptions. In Section 3 we relate all three manifolds to spaces of analytic functions, and following Kirillov and Yur'ev~\cite{Unirreps, KYu}, give a description of the Virasoro generators as vectors of the tangent space to the space of smooth univalent functions at an arbitrary point. A brief definition of the Segal-Wilson Grassmannian $\Gr$ and of the smooth
Grassmannian $\Gr_{\infty}$ is given in Section 4. We provide some necessary background of L\"owner-Kufarev smooth evolution in Section 5. Then in Section 6, we construct Hamiltonian formalism for the L\"owner-Kufarev evolution and define the Poissoon structure. The main result is contained in Section 7 where we construct the embedding of the L\"owner-Kufarev evolution into the Segal-Wilson Grassmannian.  We prove that
the Virasoro generators in their covariant form are conserved along the Hamiltonian flow. Then we present the $\tau$-function in Section 8. Section 9 gives the relation of the shape evolution to integrable systems. We construct the Baker-Akhiezer function, define the KP flows, and finally, we find explicitly a class of potentials in the Lax operator, which  satisfy the KP equation.

\vspace{10pt}

\noindent
{\bf Acknowledgement.} The authors are thankful for the support of NILS mobility project and prof. Fernando P\'erez-Gonz\'alez (Universidad de La Laguna, Espa\~na), the University of Chicago and prof. Paul Wiegmann  for support and helpful discussions. The authors acknowledge also helpful discussions with prof. Roland Friedrich during his visit to the University of Bergen in 2008.

\section{Definitions and structures of $Vir$, $\Diff S^1$,  and $\Diff S^1/S^1$}\label{GrAl}

\subsection{Witt and Virasoro algebras}\label{WV}

The complex {\it Witt algebra} is the Lie algebra of holomorphic vector fields defined on $\mathbb C^*=\mathbb C\setminus\{0\}$ acting by derivation
over the ring of Laurent polynomials $\mathbb C[z,z^{-1}]$. It is spanned by the basis ${L}_n=z^{n+1}\frac{\partial}{\partial z}$, $n\in \mathbb Z$.
The Lie bracket of two basis vector fields is given by the commutator $[L_n,L_m]=(m-n)L_{n+m}$. Its central extension is the complex {\it Virasoro algebra} $\mathfrak{vir}_{\mathbb C}$ with the central element $c$ commuting with all $L_n$, $[L_n,c]=0$, and with the Virasoro commutation relation
\[
[L_n,L_m]=(m-n)L_{n+m}+\frac{c}{12}n(n^2-1)\delta_{n,-m}, \quad n,m\in\mathbb Z,
\]
where $c\in \mathbb C$ is the central charge denoted by the same character. These algebras play important role in conformal field theory. In order to construct their representations   one can use an analytic realization.

\subsection{Group of diffeomorphisms}

Let us denote by $\Diff S^1$ the group of orientation preserving $C^{\infty}$ diffeomorphisms
of the unit circle $S^1$, where the group operation is given by the superposition of diffeomorphisms, the identity element of the group is the identity map on the circle, and the inverse element is the inverse diffeomorphism. Topologically the group $\Diff S^1$ is an open subset of the space of smooth functions on the unit circle $C^{\infty}(S^1\to S^1)$, endowed with the $C^{\infty}$-topology. This allows us to consider the group $\Diff S^1$ as a Lie-Frech\'et group.
The corresponding Lie-Frech\'et algebra $\mathfrak{diff} \,\, S^1$ is identified with the tangent space $T_{\id}\Diff S^1$ at the identity~$\id$, and it inherits the Frech\'et topology from $C^{\infty}(S^1\to S^1)$. In its turn $T_{\id}\Diff S^1$ can be thought of as the set of all velocity vectors of smooth curves at time zero passing through~$\id$. Every such velocity vector is just a smooth real vector field on $S^1$. Denote by $\Vect S^1=\{\phi=\phi(\theta)\frac{d}{d\theta}\mid \phi\in C^{\infty}(S^1\to\mathbb R)\}$ the space of smooth real vector fields on the circle. This construction allows us to identify the Lie-Frech\'et algebra $\mathfrak{diff} \,\, S^1$ of $\Diff S^1$ with the space $\Vect S^1$ equipped with the Lie brackets $[\phi_1(\theta)\frac{d}{d\theta},\phi_2(\theta)\frac{d}{d\theta}]$, see e.g.;~\cite{Milnor}.

The Virasoro-Bott group $Vir$ is the central extension of the group $\Diff S^1$ by the group of real numbers $\mathbb R$. This central extension is given by the Bott continuous cocycle~\cite{Bott}, which is a map $\Diff S^1\times \Diff S^1\to S^1$ of the form
\[
(\varphi_1,\varphi_2)\to \frac{1}{2}\int_{S^1}\log (\varphi_1\circ\varphi_2)' d\log \varphi_2'.
\]
The Lie algebra $\mathfrak{vir}$ for $Vir$ is called the (real) Virasoro algebra and it is given by the central extension of the Lie-Frech\'et algebra $\Vect S^1$ by the algebra of real numbers. The central extension is unique nontrivial modulo isomorphisms and is given by the Gelfand-Fuchs 2-cocycle~\cite{GelfandFuchs}
\[
\omega(\phi_1,\phi_2)=\int_{S_1}\phi_1'(\theta)\phi_2''(\theta)d\theta.
\]
Both groups $\Diff S^1$ and $Vir$ are modeled over a real Fr\'echet space.

Let us denote by $[\id]$ the equivalence class in $\Diff S^1/S^1$ of the identity element $\id\in \Diff S^1$. Then $T_{[\id]}\Diff S^1/S^1$ is associated
with the quotient $\Vect_0 S^1=\Vect S^1/const$ of the algebra $\Vect S^1$ by the constant vector fields and can be realized as the space of vector fields $\phi(\theta)\frac{d}{d\theta}$ from  $\Vect S^1$ with vanishing mean value over~$S^1$. All constant vector fields form the equivalence class $[0]$.

\subsection{CR and complex structures}\label{CR}

In Section~\ref{relations} we shall describe relations between the groups $Vir$, $\Diff S^1$, the homogeneous manifold $\Diff S^1/S^1$ and different spaces of univalent functions. The algebraic objects are essentially real, meanwhile the spaces of univalent functions carry natural complex structures as well as the  algebraic definition of the Witt and Virasoro algebras in Subsection~\ref{WV} considers vector fields over the field of complex numbers. Therefore, we need to complexify the real objects in order to present these relations.
Structures and mappings on infinite dimensional manifolds are more general than for finite-dimensional ones, however, being restricted to the latter they coincide with the standard ones. For the completeness we give some necessary definitions mostly based on~\cite{Boggess, Lempert}.

Given a smooth manifold $\mathcal M$,  we consider  the tangent space $T_p\mathcal M$ at each point $p\in \mathcal M$ as a real vector space. After tensoring with $\mathbb C$ and splitting $T_p\mathcal M\otimes\mathbb C=T_p^{(1,0)}\mathcal M\oplus T_p^{(0,1)}\mathcal M$, we form the holomorphic $T^{(1,0)}\mathcal M$ and antiholomorphic $T^{(0,1)}\mathcal M$ tangent bundles. The pair $(\mathcal M,T^{(1,0)}\mathcal M)$ is an almost complex manifold which becomes complex in the integrable case meaning that any commutator of vector fields from $T^{(1,0)}\mathcal M$ remains in $T^{(1,0)}\mathcal M$, and similarly, the commutators of vector fields from $T^{(0,1)}\mathcal M$ remain in $T^{(0,1)}\mathcal M$.

A Lie group $\mathbb G$ with a neutral element $e$ and with a Lie algebra $\mathfrak g$ possesses a left invariant complex structure $(\mathbb G, \mathfrak g^{(1,0)})$ if one can construct a complexification $\mathfrak{g}_{\mathbb C}=(T_{e}\mathbb G)_{\mathbb C}$ of the Lie algebra $\mathfrak g$, such that the decomposition $\mathfrak{g}_{\mathbb C}=\mathfrak g^{(1,0)}\oplus\mathfrak g^{(0,1)}$ is integrable, that is equivalent to say that $\mathfrak g^{(1,0)}$ is a subalgebra.

Let us recall the definition of the Cauchy-Riemann (CR) structure on a manifold $\mathcal N$. Given a smooth manifold $\mathcal N$ and its complexified tangent bundle
$T\mathcal N\otimes \mathbb C$ we find a complex corank one subbundle $H$ of $T\mathcal N\otimes \mathbb C$. The splitting $H=H^{(1,0)}\oplus H^{(0,1)}$ defines an almost CR structure. If it is integrable, then the pair $(\mathcal N, H^{(1,0)})$ is called a CR manifold. Roughly speaking the holomorphic part of a CR structure represents a maximal subbundle of the real tangent bundle that admits a complex structure.
The left-invariant CR-structure $(\mathbb G, \mathfrak h^{(1,0)})$ is defined similarly to the left-invariant complex structure above.

As an example of CR manifold we can consider an embedded real hypersurface (that is an embedded real corank 1 submanifold) into a complex manifold. Namely, let $\mathcal N$ be a real hypersurface of the complex manifold $(\mathcal M, T^{(1,0)}\mathcal M)$. Then the CR manifold $(\mathcal N, H^{(1,0)})$ is defined by setting $H^{(1,0)}=T^{(1,0)}\mathcal M\Big|_{\mathcal N}\bigcap (T\mathcal N\otimes \mathbb C)$.

A CR manifold $(\mathcal N, H^{(1,0)})$ is called pseudoconvex if $[X,\bar X]\notin H^{(1,0)}\oplus H^{(0,1)}$ for any non-vanishing vector field $X\in H^{(1,0)}$.

A smooth mapping  $F$ from a complex manifold $(\mathcal M_1, T^{(1,0)}\mathcal M_1)$ to a complex manifold $(\mathcal M_2, T^{(1,0)}\mathcal M_2)$ is called holomorphic if the holomorphic part $\partial F$ of its differential $dF= \partial F+\bar{\partial} F$ is the mapping $\partial F: T^{(1,0)}\mathcal M_1\to T^{(1,0)}\mathcal M_2$ and $\bar\partial F=0$.
The problem of solving the equation $\bar\partial F=0$ is quite difficult. Some of  results in this direction are found, e.g.;~\cite{Raboin}.

Analogously, a smooth mapping $F$  from a CR manifold $(\mathcal N_1, H^{(1,0)}_1)$ to a CR manifold  $(\mathcal N_2, H^{(1,0)}_2)$ is called CR if  its holomorphic differential is a  map $\partial F: H^{(1,0)}_1\to H^{(1,0)}_2$ and $\bar\partial F=0$.

Given a non-trivial representative $\phi$ of the equivalence class $[\phi]$ of  $\Vect_0S^1$
\[
\phi(\theta)=\sum\limits_{n=1}^{\infty}a_n\cos\,n\theta+b_n\sin\,n\theta,
\]
let us define an almost complex structure $J$ by the operator
\[
J(\phi)(\theta)=\sum\limits_{n=1}^{\infty}-a_n\sin\,n\theta+b_n\cos\,n\theta.
\]
Then $J^2=-id$. On ${\Vect_0}_{\mathbb C}:=\Vect_0S^1\otimes \mathbb C$, the operator $J$ diagonalizes and we have the isomorphism
\[
\Vect_0S^1\ni \phi\leftrightarrow v:=\frac{1}{2}(\phi-iJ(\phi))=\sum\limits_{n=1}^{\infty}(a_n-ib_n)e^{in\theta}\in H^{(1,0)}:= (\Vect_0S^1\otimes \mathbb C)^{(1,0)},
\]
and the latter series extends into the unit disk as a holomorphic function. So $\Diff_{\mathbb C}S^1/S^1= (\Diff S^1/S^1, H^{(1,0)})$ becomes a complex manifold and $(\Diff S^1, H^{(1,0)})$ becomes a CR manifold where $H^{(1,0)}$ is isomorphic to $\Vect_0 S^1$. Thus, the group $\Diff S^1$ possesses the left-invariant CR-structure $(\Diff S^1, H^{(1,0)})$, and $C^*$ forms a Cartan subalgebra of $\Vect S^1\otimes \mathbb C= (H^{(1,0)}\oplus  H^{(0,1)})\oplus \mathbb C^*$. Taking the complex Fourier basis $v_n=e^{in\theta}\frac{d}{d\theta}$, $n\in \mathbb Z$, in $\Vect S^1\otimes \mathbb C$ we arrive at the Witt commutation relations $[v_n, v_m]=(m-n)v_{n+m}$, where the commutators $[v_n, v_m]$ remain in $H^{(1,0)}$ for $n,m>0$ and in $H^{(0,1)}$ for $n,m<0$, however the Lie hull Lie$(H^{(1,0)},H^{(0,1)})\not\subset H^{(1,0)}\oplus H^{(0,1)}$.

\section{Relations between $Vir$,  $\Diff S^1$, and $\Diff S^1/S^1$ and spaces of univalent functions}\label{relations}

Let us introduce necessary classes of univalent functions in order to formulate main statements.
Let $\mathcal A_0$ and $\widetilde{\mathcal A}_0$ denote the classes of holomorphic functions in the unit disk $\mathbb D$ defined by$$
\mathcal A_0=\{f\in C^{\infty}(\hat{\mathbb D})\ \mid\ f\in \Hol(\mathbb D),\  f(0)=0\},\qquad \widetilde{\mathcal A}_0=\{f\in \mathcal A_0\mid\   f^{\prime}(0)=0\},
$$
where $\hat{\mathbb D}$ is the closure of the unit disk $\mathbb D$.
The classes $\mathcal A_0$ and $\widetilde{\mathcal A}_0$ are complex Frech\'et vector spaces, where the topology is defined by the seminorms
$$
\|f\|_m=\sup \{ |f^{(m)}(z)| \ \mid\ z\in\hat{\mathbb D}\},
$$ which is equivalent to the uniform convergence of all derivatives in $\hat{\mathbb D}$. Notice that both $\mathcal A_0$ and $\widetilde{\mathcal A}_0$ can be considered as complex manifolds where the real tangent space naturally isomorphic to the holomorphic part of the splitting.
Then we define
$$
\mathcal F= \{ f\in A_0\ \mid\ f\ \ \text{is univalent in}\ \ \mathbb D,\ \ \text {injective and smooth on the boundary}\ \ \partial{\mathbb D}\}.
$$
Geometrically, class $\mathcal F$ defines all differentiable embeddings of the closed disk $\hat{\mathbb D}$ to $\mathbb C$ and analytically it is represented by functions $f=cz(1+\sum_{n=1}^{\infty}c_nz^n)$, $c,c_n\in\mathbb C$. As a subset of $\mathcal A_0$ the space of univalent functions $\mathcal F$ forms an open subset inheriting the Frech\'et topology of complex vector space~$\mathcal A_0$.
Next we consider the class
$$
\mathcal F_1=\{f\in \mathcal F\ \mid\ |f^{\prime}(0)|=1\},
$$ whose elements can be written as $f=e^{i\phi}z(1+\sum_{n=1}^{\infty}c_nz^n)$, $\phi\in\mathbb R\mod 2\pi$. The set $\mathcal F_1$ is the pseudo-convex surface of real codimension $1$ in the complex open set $\mathcal F\subset \mathcal A_0$.

The last class of functions is
$$
\mathcal F_0=\{f\in \mathcal F\ \mid\ f^{\prime}(0)=1\}.
$$
The elements of this class have the form $f=z(1+\sum_{n=1}^{\infty}c_nz^n)$. It is obvious that $\mathcal F_0$ can be considered both as the quotient $\mathcal F_1/S^1$ and as the quotient $\mathcal F/\mathbb C^*$, $\mathbb C^*=\mathbb C\setminus\{0\}$. In the latter case, $\mathcal F$ is the holomorphic trivial $\mathbb C^*$-principal bundle over the base space $\mathcal F_0$.
Since the set $\mathcal F_0$ can be also considered as an open subset of the affine space $v+\widetilde{\mathcal A}_0$, where $v(z)=z$, the tangent space $T_f\mathcal F_0$ inherits the natural complex structure of complex vector space~$\widetilde{\mathcal A}_0$~\cite{AM}. The tangent space $T_f\mathcal F_0$ with the induced complex structure from $\widetilde{\mathcal A}_0$ is isomorphic to the complex vector space $T_f^{(1,0)}\mathcal F_0$ of the complexification
$T\mathcal F_0\otimes \mathbb  C=T^{(1,0)}\mathcal F_0\oplus T^{(0,1)}\mathcal F_0$. Moreover, the affine coordinates can be introduced so that  to every $f\in\mathcal F_0$, written in the form $f(z)=z(1+\sum_{n=1}^{\infty}c_nz^n)$ there will correspond the sequence $\{c_n\}_{n=1}^{\infty}$.

\begin{theorem}\cite{Lempert}\label{vir}
The Virasoro-Bott group $Vir$ has a left invariant complex structure, and as a complex manifold $Vir_{\mathbb C}$, it is biholomorphic to $\mathcal F$.
\end{theorem}

\begin{theorem}\cite{Lempert}
The group $\Diff S^1$ has a left invariant CR structure and with this CR structure it is isomorphic to the hypersurface $\mathcal F_1$.
\end{theorem}

The last theorem concerns with the homogeneous space $\Diff S^1/S^1$, where $S^1$ is considered as a subgroup of $\Diff S^1$. The group $S^1$ acts transversally to CR structure of $\Diff S^1$, leaving it invariant. 

\begin{theorem}\label{Th3}\cite{KYu,Lempert}
The homogeneous space $\Diff S^1/S^1$ has a complex structure, and as a complex manifold $\Diff_{\mathbb C} S^1/S^1$,  is biholomorphic to $\mathcal F_0$.
\end{theorem}

It can be shown that $\Diff S^1/S^1$ admits not only complex but even K\"ahlerian structure.
Entire necessary background for the construction of the theory of unitary representations of $\Diff S^1$ is found in~\cite{AM,KYu}.

It was mentioned that $\mathcal F$ is the holomorphic trivial $\mathbb C^*$-principal bundle over $\mathcal F_0$. In order to prove Theorem~\ref{vir}, Lempert showed \cite{Lempert} that the complexification $Vir_{\mathbb C}$ of the Virasoro-Bott group $Vir$  is also a holomorphic trivial $\mathbb C^*$-principal bundle over $\Diff_{\mathbb C} S^1/S^1$. This implies the existence of a biholomorphic map between $\mathcal F$ and~$Vir_{\mathbb C}$.

We will assign the same character $\mathcal F_0$ to both, the class of univalent functions defined in the closure unit disk   $\mathcal F_0(\hat{\mathbb D})$, and the class of functions restricted to the unit circle  $\mathcal F_0(S^1)$. Obviously both classes are isomorphic. 

The right action of the group $\Diff S^1$ over the manifold $\Diff S^1/S^1$ is well defined and it gives the right action
$\Diff S^1$ over the class  $\mathcal F_0(S^1)$ due to Theorem~\ref{Th3}, which is technically impossible to write explicitly because
the Riemann mapping theorem gives no explicit formulas. However, it is possible \cite{KYu} to write the infinitesimal generator making use of the Schaeffer and Spencer variation \cite[page 32]{SS}
$$
L[f,\phi](z):=\frac{f^2(z)}{2\pi
}\int\limits_{S^1}\left(\frac{w f'(w)}{f(w)}\right)^2\frac{ \phi(w)\,dw}{w(f(w)-f(z))}\in T_{f}\mathcal F_0,
$$
defined for $f\in\mathcal F_0$, $\phi\in\Vect S^1$. It extends by linearity to a map $L[f,\cdot]:\,\Vect_{\mathbb C}S^1\to T_f\mathcal F_0\otimes \mathbb C=T^{(1,0)}_f\mathcal F_0\oplus T^{(0,1)}_f\mathcal F_0$. The variation $L[f,\cdot]$ defines  the isomorphism of vector spaces
$H^{(1,0)}\leftrightarrow T^{(1,0)}_f\mathcal F_0$, which is given explicitly by \eqref{iso}. At the same time $L[f,\cdot]$ defines an isomorphism of the Lie algebras $H^{(1,0)}\leftrightarrow T^{(1,0)}_f\mathcal F_0$, where $H^{(1,0)}$ is considered as a  subalgebra of the Witt algebra $\Vect_{\mathbb C} S^1$ and  $T^{(1,0)}_f\mathcal F_0$ is endowed with the usual commutator of vectors. In order to obtain a homomorphism of the entire Witt algebra we 
 extend $L[f,\cdot]$ to $H^{(1,0)}\oplus H^{(0,1)}\oplus \mathbb C^*\to T^{(1,0)}_f\mathcal F_0$. 
 
 Explicitly, this homomorphism  $L[f,\cdot]$ is given by the residue calculus, see e.g.;~\cite{AM, Unirreps}.
Taking the holomorphic part of the Fourier basis ${v_k}=-iz^k$, $k=1,2,\dots$, for $\Vect S^1\otimes \mathbb C$, we obtain
\begin{equation}\label{iso}
L[f,v_k](z)=L_k[f](z)=z^{k+1}f'(z)\quad L_k[f]\in T_f^{(1,0)}\mathcal F_0,
\end{equation}
and taking the antiholomorphic part of the basis $v_{-k}=-iz^{-k}$, $k=1,2,\dots$, we obtain expressions for $L_{-k}[f]\in T_f^{(1,0)}\mathcal F_0$ that are rather difficult. The first two of them are
$$
L_{-1}[f](z)=f'(z)-2c_1 f(z)-1,\qquad L_{-2}[f](z)=\frac{f'(z)}{z}-\frac{1}{f(z)}-3c_1+(c_1^2-4c_2)f(z),
$$
and others can be obtained by the commutation relations~\cite{AM, KYu}
\begin{equation}\label{WittBr}
[L_k,L_n]=(n-k)L_{k+n}, \quad k,n\in \mathbb Z.
\end{equation}
The constant vector $v_0=-i$ is mapped to $L_0[f](z)=zf^{\prime}(z)-f(z)$. The vector fields $L_k$, $k\in\mathbb Z$ were obtained in~\cite{KYu} and received the name of Kirillov's vector fields, see also~\cite{AM}.
We have
$$
T^{(1,0)}_{\id}\mathcal F_0=\spn\{L_0[id], L_1[id], L_2[id],\dots\}=\spn\{z^2,z^3,\ldots\}.
$$
Let us recall that $\id\in \mathcal F_0$ is the image of an equivalence class of the identity diffeomorphism from $\Diff S^1/S^1$.

Summarizing, 
the group $\Diff S^1$ acts transitively on the homogeneous manifold
$\Diff S^1/S^1$ defining an action on the manifold $\mathcal F_0$. 
The infinitesimal generator of this action  produces left-invariant section
of the tangent bundle $T\mathcal F_0$ by the Schaeffer-Spencer linear map. We get the isomorphism
$$
T_f\mathcal F_0\simeq T_f^{(1,0)}\mathcal F_0=\spn\{L_1[f],L_2[f],\dots\},
$$
at a point $f\in\mathcal F_0$. The vector $L_0[f]$ is the image of the constant unit vector $i$ under the Schaeffer-Spencer linear map at an arbitrary point $f\in\mathcal F_0$ with value $0$ at $\id\in\mathcal F_0$.

The vector fields $L_k$, $k\in\mathbb Z$, at $f(z)=z\left(1+\sum_{n=1}^{\infty}c_nz^n\right)\in\mathcal F_0$ can be written in the affine coordinates $\{c_n\}_{n=1}^{\infty}$ by making use of the isomorphism $z^{n+1}\mapsto\partial_n$, where $\partial_n=\frac{\partial}{\partial c_n}$ as the following first order differential operators
$$
L_k[f]=\partial_k+\sum_{n=1}^{\infty}(n+1)c_n\partial_{k+n},\quad k>0,
$$
\begin{equation}\label{affine}
L_0[f]=\sum_{n=1}^{\infty}nc_n\partial_n,\qquad L_{-1}[f]=\sum_{n=1}^{\infty}\Big((n+2)c_{n+1}-2c_1c_n\Big)\partial_n,
\end{equation}
$$
 L_{-2}[f]=\sum_{n=1}^{\infty}\Big((n+3)c_{n+2}+(c_1^2-4c_2)c_{n}-\alpha_{(n+2)}\Big)\partial_n,
$$
where $\alpha_n$ can be found from the recurrent relations $\alpha_n=-\sum_{k=1}^{n}c_k\alpha_{n-k}$, $\alpha_0=1$. Here, for example,
$$\alpha_1=-c_1,\quad \alpha_2=c_1^2-c_2,\quad \alpha_3=-c_1^3+2c_1c_2-c_3,\quad\ldots.$$
For other negative values of $k$ the expressions of $L_k[f]$ are more complicated but can be found by an algebraic procedure, see e.g.;~\cite{AM, AiraultNeretin}.

\section{Segal-Wilson Grassmannian}

Sato's (universial) Grassmannian appeared first in 1982 in~\cite{Sato}  as an infinite dimensional generalization of the classical finite dimensional Grassmannian manifolds and they are described as `the topological closure of the inductive limit of' a finite dimensional Grassmanian as the dimensions of the ambient vector space and its subspaces tend to infinity. It turned out to be a very important infinite dimensional manifold being related to the representation theory of loop groups, integrable hierarchies, micrological analysis, conformal and quantum field theories, the second quantization of fermions, and to many other topics~\cite{DJKM83,MPPR,SegalWilson,Witten}.
In the Segal and Wilson
approach~\cite{SegalWilson} the infinite dimensional Grassmannian $\Gr(H)$ is taken over the separable Hilbert space $H$. The first systematic description of the  infinite dimensional Grassmannian can be found in~\cite{PrSe}.

We present here a general definition of the infinite dimensional smooth Grassmannian $\Gr_{\infty}(H)$.
As a separable Hilbert space  we take the space $L^2(S^1)$ and consider its dense subspace $H=C^{\infty}_{\|\cdot\|_2}(S^1)$  of smooth complex valued functions defined on the unit circle endowed with  $L^2(S^1)$ inner product $\langle f,g\rangle=\frac{1}{2\pi}\int\limits_{S^1}f\bar g\,dw$, $f,g\in H$. The orthonormal basis of $H$ is $\{z^k\}_{k\in\mathbb Z}=\{e^{ik \theta}\}_{k\in\mathbb Z}$, $e^{i\theta} \in S^1$. 
Let us split all integers $\mathbb Z$ into two sets $\mathbb Z^+=\{0,1,2,3,\dots\}$ and $\mathbb Z^-=\{\dots,-3,-2,-1\}$, and
let us define a polarization by  $$H_+=\spn_H\{z^k,\ k\in\mathbb Z^+\}, \qquad H_-=\spn_H\{z^k,\ k\in \mathbb Z^-\}.$$
Here and further span is taken in the appropriate space indicated as a subscription. The Grassmanian is thought of as the set of closed linear subspaces $W$ of $H$, which are commensurable with $H_+$ in the sense that they have finite codimension in both $H_+$ and~$W$. This can be defined by means of the descriptions of the orthogonal projections of the subspace $W\subset H$ to $H_+$ and $H_-$.

\begin{definition}
The infinite dimensional smooth Grassmannian $\Gr_{\infty}(H)$ over  the space $H$  is the set of  subspaces $W$ of $H$, such that
\begin{itemize}
\item[1.]{the orthogonal projection $pr_+\colon W\to H_+$ is a Fredholm operator,}
\item[2.]{the orthogonal projection $pr_-\colon W\to H_-$ is a compact operator.}
\end{itemize}
\end{definition}

The requirement that $pr_+$ is Fredholm means that the kernel and cokernel of $pr_+$ are finite dimensional. More information about Fredholm operators the reader can find in~\cite{Doug}. It was proved in~\cite{PrSe}, that $\Gr_{\infty}(H)$ is a dense submanifold in a Hilbert manifold modeled over the space $\mathcal L_2(H_+,H_-)$ of Hilbert-Schmidt operators from $H_+$ to $H_-$, that  itself has the structure of a Hilbert space, see~\cite{RS}.  Any $W\in\Gr_{\infty}(H)$ can be thought of as a graph $W_T$ of a Hilbert-Schmidt operator $T\colon W\to W^{\bot}$, and  points of a neighborhood $U_W$ of $W\in \Gr_{\infty}(H)$ are in one-to-one correspondence with operators from $\mathcal L_2(W,W^{\bot})$.

Let us denote by $\mathfrak S$ the set of all collections $\mathbb S\subset \mathbb Z$ of integers such that $\mathbb S\setminus \mathbb Z^+$ and $\mathbb Z^+\setminus \mathbb S$ are finite. Thus, any sequence $\mathbb S$ of integers is bounded from below and contains all positive numbers starting from some number. It is clear that the sets $H_{\mathbb S}=\spn_{H}\{z^k,\ k\in\mathbb S\}$ are elements of the Grassmanian $\Gr_{\infty}(H)$ and they are usually called {\it special points}. The collection of neighborhoods $\{U_{\mathbb S}\}_{\mathbb S\in \mathfrak S},$
$$
U_{\mathbb S}=\{W\ \mid\ \text{there is an orthogonal projection}\ \pi\colon  W\to H_{\mathbb S}\ \text{that is an isomorphism}\}
$$
forms an open cover of $\Gr_{\infty}(H)$. The virtual cardinality of $\mathbb S$ defines the {\it virtual dimension} (v.d.) of $H_{\mathbb S}$, namely:
\begin{equation}\label{virtdim}
\virtcard(\mathbb S)  = \virtdim(H_{\mathbb S})=\dim(\mathbb N\setminus\mathbb S)-\dim(\mathbb S\setminus\mathbb N)
 = \ind(pr_+).
\end{equation}
The expression $\ind(pr_+)=\dim\kernel(pr_+)-\dim\cokernel(pr_-)$ is called the index of the Fredholm operator~$pr_+$. According to their virtual dimensions the points of $\Gr_{\infty}(H)$ belong to different components of connectivity. The  Grassmannian is the disjoint union of connected components parametrized by their virtual dimensions.

\section{L\"owner-Kufarev evolution}

The pioneering idea of L\"owner \cite{Loewner} in 1923 contained two main ingredients: subordination chains and semigroups of conformal maps. This far-reaching program was created in the hopes to solve
the Bieberbach conjecture \cite{Bieb} and the final proof of this conjecture by de~Branges \cite{Branges} in 1984 was based on L\"owner's  parametric method.
The modern form of this method is due to Kufarev \cite{Kufarev} and Pommerenke \cite{Pommerenke1, Pommerenke2}. Omitting review over subordination chains we concentrate our attention on the other ingredient, i.e.;  on evolution families relating them to semigroups as in  \cite{Contreras, Goryainov, Pommerenke2}.

Let us consider a semigroup $\mathcal P$ of conformal univalent maps from the unit disk $\mathbb D$ into itself with superposition as a semigroup operation. This makes $\mathcal P$ a topological semigroup with respect of the topology of local uniform convergence on $\mathbb D$. We impose the natural normalization for such conformal maps $\Phi(z)=b_1z+b_2z^2+\dots$ about the origin, $b_1>0$. The unity of this semigroup is the identity map. A continuous homomorphism from $\mathbb R^+$ to $\mathcal P$ with a parameter $\tau\in \mathbb R^+$ gives a {\it semiflow} $\{\Phi^{\tau}\}_{\tau\in\mathbb R^+}\subset \mathcal P$ of conformal maps $\Phi^{\tau}:\mathbb D\to \Omega\subset \mathbb D$, satisfying the properties
\begin{itemize}
\item $\Phi^0=id$;
\item $\Phi^{\tau+s}=\Phi^{s}\circ \Phi^{\tau}$;
\item $\Phi^{\tau}(z)\to z$ locally uniformly in $\mathbb D$ as $\tau\to 0$.
\end{itemize}
In particular, $\Phi^{\tau}(z)=b_1(\tau)z+b_2(\tau)z^2+\dots$, and $b_1(0)=1$. This semi-flow is generated by a vector field $v(z)$ if for each $z\in \mathbb D$ the function $w=\Phi^{\tau}(z), \tau\geq 0$ is a solution to an autonomous differential equation $dw/d\tau=v(w)$, with the initial condition $w(z,\tau)\bigg|_{\tau=0}=z$. This vector field, called  infinitesimal generator, is given by $v(z)=-zp(z)$ where $p(z)$ is a regular Carath\'eodory function in the unit disk, with $\re p(z)>0$ in $\mathbb D$.

We call a subset $\Phi^{t,s}$ of $\mathcal P$, $0\leq s\leq t$ an {\it evolution family} if
\begin{itemize}
\item $\Phi^{t,t}=id$;
\item $\Phi^{t,s}=\Phi^{t,r}\circ \Phi^{r,s}$, for $0\leq s\leq r\leq t$;
\item $\Phi^{t,s}(z)\to z$ locally uniformly in $\mathbb D$ as $t-s\to 0$.
\end{itemize}
In particular, if $\Phi^{\tau}$ is a one-parameter semiflow, then $\Phi^{t-s}$ is an evolution family. Given an evolution family $\{\Phi^{t,s}\}_{t,s}$, every function $\Phi^{t,s}$ is univalent and there exists an essentially unique infinitesimal generator, called the Herglotz vector field $H(z,t)$, such that
\begin{equation}\label{LK1}
\frac{d\Phi^{t,s}(z)}{d t}=H(\Phi^{t,s}(z),t),
\end{equation}
where the function $H$ is given by $H(z,t)=-zp(z,t)$ with a Carath\'eodory function $p$ for almost all $t\geq 0$. The converse is also true. Solving equation \eqref{LK1} with the initial condition $\Phi^{s,s}=\id$, we obtain an evolution family. In particular, we can consider situation when $s=0$. Let
$f(z,t)=e^tw(z,t)$. A remarkable property
of evolution families is that any conformal embedding $f$ of the unit disk $\mathbb D$ to $\mathbb C$ normalized by $f(z)=z+c_1z^2+\dots$ in $\mathbb D$ can be obtained as a one-parameter homotopy
from the identity map, i.e.;
$$
f(z)=\lim\limits_{t\to\infty} f(z,t)=\lim\limits_{t\to\infty} e^{t}w(z,t),
$$
where the function
\[
w(z,t)=e^{-t}z\left(1+\sum\limits_{n=1}^{\infty}c_n(t)z^n\right),
\]
solves the Cauchy problem for the L\"owner-Kufarev ODE
\begin{equation}\label{LK2}
\frac{dw}{dt}=-wp(w,t),\quad w(z,t)\bigg|_{t=0}=z,
\end{equation}
and with the function
$p(z,t)=1+p_1(t)z+\dots$ which is holomorphic in $\mathbb D$ for almos all $t\in [0,\infty)$, measurable with respect to
$t\in [0,\infty)$ for any fixed $z\in\mathbb D$, and such that $\re p>0$ in $\mathbb D$, see  \cite{Pommerenke2}. The function
$w(z,t)=\Phi^{t,0}(z)$ is univalent and maps $\mathbb D$ into $\mathbb D$.   

\begin{lemma}\label{CarClass}
Let the function $w(z,t)$ be a solution to the Cauchy problem~\eqref{LK2}.
 If the driving function $p(\cdot,t)$, being from the Carath\'eodory class for almost all $t\geq 0$, is  $C^{\infty}$ smooth in the closure $\hat{\mathbb D}$ of the unit disk $\mathbb D$ and summable with respect to $t$, then the boundaries of
the domains $\Omega(t)=w(\mathbb D,t)\subset \mathbb D$ are smooth for all $t$ and $w(\cdot,t)$ extended to $S^1$ is injective on $S^1$.
\end{lemma}
\begin{proof}
Observe that the continuous and differentiable dependence of
 the solution of a differential equation $\dot{x}=F(t,x)$ on the initial condition $x(0)=x_0$ is a classical problem. One can refer, e.g., to \cite{Volpato}
in order to assure that summability of $F(\cdot, x)$ regarding to $t$ for each fixed $x$ and  continuous differentiability ($C^1$ with respect to
$x$ for almost all $t$) imply that the solution $x(t,x_0)$ exists, is unique, and is $C^1$ with respect to $x_0$. In our case, the solution to  \eqref{LK2}
exists, is unique and analytic in $\mathbb D$, and moreover, $C^1$ on its boundary $S^1$. Let us differentiate \eqref{LK2} inside the unit disk $\mathbb D$ with respect to $z$ and write
\[
\log w' =-\int\limits_{0}^{t}(p(w(z,\tau),\tau)+w(z,\tau)p'(w(z,\tau),\tau))d\tau,
\]
choosing the branch of the logarithm such as $\log w'(0,t)=-t$.
This equality is extendable onto $S^1$ because the right-hand side is, and therefore, $w'\in C^1(S^1)$ and $w\in C^2(S^1)$. We continue analogously and write the formula
\[
w''=-w'\int\limits_{0}^{t}(2w'(z,\tau)p'(w(z,\tau),\tau)
+w(z,\tau)w'(z,\tau)p''(w(z,\tau),\tau))d\tau,
\]
which guarantees that $w\in C^3(S^1)$. Finally, we come to the conclusion that $w$ is $C^\infty$ on $S^1$.
\end{proof}

Let us denote by $f(z,\infty)$ the final point of the trajectory $f(z,t)=e^tw(z,t)$, $t\in [0,\infty)$, where  $w(z,t)$ is a solution to  the Cauchy problem~\eqref{LK2} with the driving function $p(z,t)$ satisfying the conditions of Lemma~\ref{CarClass}. Then $f(z,t)\in \mathcal F_0$ for all $t\in [0,\infty)$ but not necessarily for $t=\infty$. One can formulate a stronger reciprocal statement.

\begin{lemma}\label{lemmaX}
With the above notations let $f(z)\in \mathcal F_0$. Then there exists a function $p(\cdot,t)$ from the Carath\'eodory class for almost all $t\geq 0$, and $C^{\infty}$ smooth in  $\hat{\mathbb D}$, such that $f(z)=\lim_{t\to\infty}f(z,t)$ is the final point of the L\"owner-Kufarev trajectory with the driving term $p(z,t)$.
\end{lemma}
\begin{proof}
Indeed, the domain $\Omega^+=f(\mathbb D)$ has a complement $\Omega^-$ which is a simply connected domain with infinity $\infty$ as an internal
point of $\Omega^-$ and $\partial \Omega^+=\partial \Omega^-$. Let us construct a subordination chain $\Omega^+(t)$ such that $\partial \Omega^+(t)$ is a level line of the Green function of the domain $\Omega^-$ with a singularity at $\infty$, and such that the conformal radius of $\Omega^+(t)$ with respect to the origin is equal to $e^t$. This can be always achieved, see \cite{Pommerenke2}. Then we can construct a one-parameter subordination chain of univalent maps $F(z,t)=e^t(z+\dots)$, $F(\cdot, t):\,\,\mathbb D\to \Omega^+(t)$ that exists for the time interval $[0,\infty)$, $f(z)=F_0(z)=F(z,0)$ and $f(\mathbb D)=\Omega^+=\Omega^+(0)$, and such that $\Omega^+(\infty)=\mathbb C$. Set up the function $p(z,t)=\dot{F}/zF'$, where
$\dot{F}$ and $F'$ are the real $t$-derivative and the complex $z$-derivative respectively. It is obviously smooth on the boundary and belongs to the Carath\'eodory class. The function $w(z,t)=F(F^{-1}(z,t),0)$ is defined in the whole unit disk (as an analytic continuation from $F^{-1}(F_0(z),t)\subset \mathbb D$), satisfies the L\"owner-Kufarev equation~\eqref{LK2}, and $f(z,t)=e^tw(z,t)$ has the limit $f(z)=f(z,\infty)$.
The latter statement can be found in \cite{Pommerenke2, ProkhVas}.
\end{proof}

\section{Hamiltonian formalism}

Let the driving term $p(z,t)$ in the L\"owner-Kufarev ODE~\eqref{LK2} be from the Carath\'eodory class for almost all $t\geq 0$,   $C^{\infty}$-smooth in  $\hat{\mathbb D}$, and summable with respect to $t$ as in Lemma~\ref{CarClass}.
 Then the domains $\Omega(t)=f(\mathbb D,t)=e^t w(\mathbb D,t)$  have  smooth boundaries $\partial
 \Omega(t)$ and the function $f$ is injective on $S^1$, i.e.; $f\in \mathcal F_0$. So the L\"owner-Kufarev equation can be extended to the
 closed unit disk $\hat{\mathbb D}=\mathbb D\cup S^1$.

Let us consider  functions ${\psi}\in H=C^{\infty}_{\|\cdot\|_2}$ from $T_{f}^*\mathcal F_0\otimes\mathbb C$, $f\in\mathcal F_0$,
\[
\psi(z)=\sum\limits_{k\in\mathbb Z}\psi_kz^{k-1}, \quad |z|=1,
\]
 and the space of observables on  $T^*\mathcal F_0\otimes\mathbb C$, given by integral functionals 
\[
\mathcal R(f,\bar\psi, t)=\frac{1}{2\pi}\int_{z\in S^1}r(f(z),\bar\psi(z), t)\frac{dz}{iz},
\]
where the function $r(\xi,\eta,t)$ is smooth in variables $\xi,\eta$ and measurable in $t$.

We define a special observable, the time-dependent pseudo-Hamiltonian $\mathcal H$,  by
 \begin{equation}\label{Ham3}
 \mathcal H(f,\bar\psi,p,t)=\frac{1}{2\pi}\int_{z\in S^1}\bar{z}^2f(z,t)(1-p(e^{-t}f(z,t),t))\bar \psi(z,t)\frac{dz}{iz},
 \end{equation}
 with the driving function (control) $p(z,t)$ satisfying the above properties.
The Poisson structure on the space of observables is given by the canonical brackets
\[
\{\mathcal R_1, \mathcal R_2\}=2\pi \int_{z\in S^1}z^2\left(\frac{\delta \mathcal R_1}{\delta f}\frac{\delta \mathcal R_2}{\delta \bar \psi}-\frac{\delta \mathcal R_1}{\delta \bar \psi}\frac{\delta \mathcal R_2}{\delta f}\right)\frac{dz}{iz},
\]
where $\frac{\delta}{\delta f}$ and $\frac{\delta}{\delta \overline{\psi}}$ are the variational derivatives, $\frac{\delta}{\delta f}\mathcal R=\frac{1}{2\pi}\frac{\partial}{\partial f}r$, $\frac{\delta}{\delta \overline{\psi}}\mathcal R=\frac{1}{2\pi}\frac{\partial}{\partial \overline{\psi}}r$.

Representing the coefficients $c_n$ and $\bar\psi_m$ of $f$ and $\bar{\psi}$ as integral functionals 
\[
c_n=\frac{1}{2\pi}\int_{z\in S^1}\bar{z}^{n+1}f(z,t)\frac{dz}{iz},\quad \bar\psi_m= \frac{1}{2\pi}\int_{z\in S^1}{z}^{m-1}\bar\psi(z,t)\frac{dz}{iz}, 
\]
$n\in \mathbb N$, $m\in\mathbb Z$, we obtain
$\{c_n, \bar\psi_m\}=\delta_{n,m}$, $\{c_n, c_k\}=0$, and $\{\bar\psi_l, \bar\psi_m\}=0$, where $n, k\in \mathbb N$, $l, m\in\mathbb Z$.

The infinite-dimensional Hamiltonian system is written as
\begin{equation}\label{sys1}
\frac{d c_k}{dt}=\{c_k, \mathcal H\},
\end{equation}
\begin{equation}\label{sys2}
\frac{d\bar \psi_k}{dt}=\{\bar{\psi_k}, \mathcal H\},
\end{equation}
where $k\in \mathbb Z$ and $c_0=c_{-1}=c_{-2}=\dots=0$, or equivalently, multiplying by corresponding  powers of $z$ and summing up,
\begin{equation}\label{sys10}
\frac{d f(z,t)}{dt}=f(1-p(e^{-t}f,t))=2\pi \frac{\delta  \mathcal H}{\delta \overline{\psi}}z^2=\{f, \mathcal H\},
\end{equation}
\begin{equation}\label{sys20}
\frac{d\bar \psi}{dt}=-(1-p(e^{-t}f,t)-e^{-t}fp'(e^{-t}f,t))\bar
\psi=-2\pi \frac{\delta \mathcal H}{\delta f}z^2=\{\bar\psi, \mathcal H\},
\end{equation}
where $z\in S^1$. So the phase coordinates $(f,\bar{\psi})$ play the role of the canonical Hamiltonian pair. Observe that the equation~\eqref{sys10} is the L\"owner-Kufarev equation~\eqref{LK2} for the function $f=e^{t}w$.

Let us set up the {\it generating function} $\mathcal{G}(z)=\sum_{k\in\mathbb Z}\mathcal{G}_kz^{k-1}$, such that  $$\bar{\mathcal G}(z):=f'(z,t) \bar{\psi}(z,t).$$ Consider the `non-positive' $(\bar{\mathcal G}(z))_{\leq 0}$ and `positive' $(\bar{\mathcal G}(z))_{> 0}$ parts of the Laurent series for $\bar{\mathcal G}(z)$:
\begin{equation*}
(\bar{\mathcal G}(z))_{\leq 0} = (\bar{\psi}_1+2 c_1\bar{\psi}_2+3 c_2\bar{\psi}_3+\dots)+(\bar{\psi}_2+2 c_1\bar{\psi}_3+\dots)z^{-1}+\dots
= \sum\limits_{k=0}^{\infty}\bar{\mathcal G}_{k+1}z^{-k}.
\end{equation*}
\begin{equation*}
(\bar{\mathcal G}(z))_{> 0} = (\bar{\psi}_0+2  c_1\bar{\psi}_1+3 c_2\bar{\psi}_2+\dots)z+(\bar{\psi}_{-1}+2c_1\bar{\psi}_0+3c_2\bar{\psi}_1\dots)z^{2}+\dots
= \sum\limits_{k=1}^{\infty}\bar{\mathcal G}_{-k+1}z^{k}.
\end{equation*}

\begin{proposition}\label{timeindep}
Let the driving term $p(z,t)$ in the L\"owner-Kufarev ODE be from the Carath\'eodory class for almost all $t\geq 0$,   $C^{\infty}$-smooth in $\hat{\mathbb D}$, and summable with respect to~$t$.
The functions $\mathcal G(z)$,  $(\mathcal G(z))_{< 0}$, $(\mathcal G(z))_{\geq 0}$, and all coefficients $\mathcal G_n$ are time-independent for all $z\in S^1$.
\end{proposition}

\begin{proof}
It is sufficient to check the equality $\dot{\bar{\mathcal G}}=\{\bar{\mathcal G}, {\mathcal H}\}=0$ for the function $\mathcal G$, and then, the same holds for the coefficients of the Laurent series for $\mathcal G$. 
\end{proof}

\begin{proposition}
The conjugates  $\bar{\mathcal G}_k$, $k=1,2,\ldots$, to the coefficients of the generating function satisfy the Witt commutation relation $\{\bar{\mathcal G}_m,\bar{\mathcal G}_n\}=(n-m)\bar{\mathcal G}_{n+m}$ for $n,m\geq 1$, with respect to our Poisson structure.
\end{proposition}
\noindent
The {\it proof} is straightforward.

The isomorphism $\iota: \, \bar\psi_k\to \partial_k=\frac{\partial}{\partial c_k}$, $k>0$, is a Lie algebra isomorphism $(T^{{*}^{(0,1)}}_f\mathcal F_0, \{,\})\to (T^{(1,0)}_f\mathcal F_0, [,])$. It makes a correspondence between the conjugates $\bar{\mathcal G}_n$ of the
coefficients $\mathcal G_n$ of $(\mathcal G(z))_{\geq 0}$ at the point $(f,\bar{\psi})$ and the Kirillov vectors  $L_n[f]=\partial_n+\sum\limits_{k=1}^{\infty}(k+1)c_{k}\partial_{n+k}
$, $n\in\mathbb N$. Both satisfy the Witt commutation relations~\eqref{WittBr}.

\section{Curves in Grassmannian}

Let us recall, that the underlying  space for the universal smooth Grassmannian $\Gr_{\infty}(H)$ is $H= C^{\infty}_{\|\cdot\|_2}(S^1)$ with the canonical $L^2$ inner product of functions defined on the unit circle. Its natural polarization  
  \begin{eqnarray*}
H_+  = \spn_{H}\{1, z, z^2, z^3, \dots \}, \qquad
H_-= \spn_{H}\{z^{-1}, z^{-2}, \dots \},
 \end{eqnarray*}
was introduced before.
The pseudo-Hamiltonian $\mathcal H(f,\bar\psi,t)$ is defined for an arbitrary $\psi\in L^{2}(S^1)$, but we consider only smooth solutions of the Hamiltonian system, therefore, $\psi \in H$. We identify this  space with the dense subspace of $T^*_f\mathcal F_0\otimes\mathbb C$, $f\in\mathcal F_0$. The generating function $\mathcal G$ defines a linear map $\bar{\mathcal G}$ from the dense subspace of $T^*_f\mathcal F_0\otimes \mathbb C$ to $H$, which being written  in a matrix form becomes
\begin{equation}\label{matrix}
\renewcommand{\arraystretch}{1.4}
  \left(\begin{array}{c}
  \cdots \\
 \bar{\mathcal G}_{-2}\\
  \bar{\mathcal G}_{-1}\\
  \bar{\mathcal G}_0 \\
  \hline
  \bar{\mathcal G}_1\\
  \bar{\mathcal G}_{2}\\
   \bar{\mathcal G}_{3}\\
  \cdots\\
 \end{array}\right)
=
  \left(\begin{array}{ccccc|ccccc}
\ddots&\ddots&\ddots&\ddots&\ddots &\ddots&\ddots & \ddots & \ddots & \ddots   \\
 \cdots & 0& {\bf 1} & 2c_1 & 3c_2 & 4c_3 & 5c_4 & 6c_5 & 7c_6 & \cdots \\
  \cdots &0 & 0 & {\bf 1} &2c_1 & 3c_2 & 4c_3 & 5c_4 & 6c_5&  \cdots\\
    \cdots &0 & 0 & 0 &{\bf 1} & 2c_1 & 3c_2 & 4c_3 &  5c_4& \cdots \\
 \hline
  \cdots & 0 & 0 & 0 & 0 & {\bf 1} & 2c_1 & 3c_2 & 4c_3 &  \cdots \\
\cdots &0 & 0 & 0 & 0 & 0 & {\bf 1} & 2c_1 & 3c_2 &   \cdots \\
 \cdots & 0 &0& 0 & 0 & 0 & 0 & {\bf 1} & 2c_1 &   \cdots \\
  \ddots & \ddots &\ddots &\ddots &\ddots&\ddots & \ddots & \ddots & \ddots & \ddots  \\
 \end{array}\right)
  \left(\begin{array}{c}
  \cdots \\
  \bar\psi_{-2}\\
  \bar\psi_{-1}\\
  \bar\psi_{0} \\
  \hline
  \bar\psi_1\\
 \bar\psi_{2}\\
  \bar\psi_{3}\\
  \cdots\\
 \end{array}\right)
\end{equation}
or in the matrix block form as
\begin{equation}\label{bloc}
\renewcommand{\arraystretch}{1.4}
\left(
\begin{array}{c}
\bar{\mathcal G}_{> 0}\\
\bar{\mathcal G}_{\leq 0}
\end{array} \right)=\left(
\begin{array}{cc}
C_{1,1} & C_{1,2} \\
0 & C_{1,1}
\end{array} \right) \left(
\begin{array}{c}
\bar\psi_{> 0}\\
\bar\psi_{\leq 0}
\end{array} \right),
\end{equation}

The proof of the following proposition is obvious.
\begin{proposition}  The operator $C_{1,1}\colon H_+\to H_+$ is invertible.
\end{proposition}

The generating function also defines a map ${\mathcal G}\colon T^{*}\mathcal F_0\otimes\mathbb C\to H$ by $$T^{*}\mathcal F_0\otimes\mathbb C\ni (f(z),\psi(z))\mapsto {\mathcal G}= \bar f'(z)\psi(z)\in H.$$ Observe that any solution $\big(f(z,t),\bar\psi(z,t)\big)$ of the Hamiltonian system is mapped into a single point of the space $H$, since all $\mathcal G_k$, $k\in\mathbb Z$ are time-independent  by Proposition~\ref{timeindep}.

Consider a bundle $\pi \colon\mathcal B\to T^{*}\mathcal F_0\otimes \mathbb C$ with a typical fiber isomorphic to $\Gr_{\infty}(H)$. We are aimed at construction of  a curve $\Gamma\colon [0,T]\to \mathcal B$ that is traced by the solutions to the Hamiltonian system, or in other words, by the L\"owner-Kufarev evolution. The curve $\Gamma$ will have the form $$\Gamma(t)=\Big(f(z,t),\psi(z,t),W_{T_n}(t)\Big)$$ in the local trivialization. Here $W_{T_n}$ is the graph of a finite rank operator $T_n\colon H_+\to H_-$, such that $W_{T_n}$ belongs to the connected component of $U_{H_+}$ of virtual dimension $0$.
In other words, we build an hierarchy of finite rank operators $T_{n}\colon H_+\to H_-$,  $n\in \mathbb Z^+$, whose graphs in the neighborhood $U_{H_+}$ of the point $H_+\in \Gr_{\infty}(H)$ are
 \[
T_{n}((\mathcal G(z))_{> 0})=  T_{n}({\mathcal G}_1,{\mathcal G}_2,\dots,{\mathcal G}_k,\dots )=\left\{
 \renewcommand{\arraystretch}{1.4}
 \begin{array}{l}
 { {G}_0}({\mathcal G}_1,{\mathcal G}_2,\dots,{\mathcal G}_k,\dots)\\
  { {G}_{-1}}({\mathcal G}_1,{\mathcal G}_2,\dots,{\mathcal G}_k,\dots)\\
 \dots \\
 { {G}_{-n+1}}({\mathcal G}_1,{\mathcal G}_2,\dots,{\mathcal G}_k,\dots),
 \end{array}
 \right.
 \]
 with ${G}_0z^{-1}+{G}_{-1}z^{-2}+ \ldots +{G}_{-n+1}z^{-n}\in H_-$. Let us denote by ${G}_k={\mathcal G}_k$, $k\in \mathbb N$. The elements ${G}_0, {G}_{-1},{G}_{-2}, \dots$ are constructed so that all $\{\bar{G}_k\}_{k=-n+1}^{\infty}$ satisfy the truncated Witt commutation relations 
 \[
 \{\bar{G}_k,\bar{G}_l\}_n=
 \begin{cases}
 (l-k)\bar{G}_{k+l},\quad &\text{for $k+l\geq -n+1$},\\
 0,\quad &\text{otherwise},
 \end{cases}
 \]
and are related to the Kirilov's vector fields under the isomorphism $\iota$. The projective limit as $n\leftarrow \infty$ recovers the whole Witt algebra and the Witt commutation relations.
 We present an explicit algorithm consisting of two steps in order to define the coefficients ${G}_{-k}$, $k=0,1,2,\ldots,n-1$.

\medskip

\noindent {\sc Step 1.} In the first step we remove the dependence of $\bar{\mathcal G}_{> 0}=\big\{\bar{\mathcal G}_{-k}\big\}_{k=0}^{\infty}$ on $\bar\psi_{> 0}=\big\{\bar\psi_{-k}\big\}_{k=0}^{\infty}$ defining
\begin{equation}\label{gtilde}
\widetilde{\mathcal G}_{> 0}=\bar{\mathcal G}_{> 0}-C_{1,1}\bar\psi_{> 0},
\end{equation}
where $C_{1,1}$ is the  upper triangular block in the matrix~\eqref{bloc}.
Thus, $\widetilde{\mathcal G}_{> 0}=\widetilde{\mathcal G}_{> 0}(\bar\psi_{\leq 0})$. Since the matrix $C_{1,1}$ is invertible we can write $\bar\psi_{\leq 0}=C_{1,1}^{-1}\bar{\mathcal G}_{\leq 0}$, that implies
$$
\widetilde{\mathcal G}_{> 0}=\widetilde{\mathcal G}_{> 0}(C_{1,1}^{-1}\bar{\mathcal G}_{\leq 0})=\widetilde{\mathcal G}_{> 0}(\bar{\mathcal G}_{\leq 0}).
$$
Let us denote by $\widetilde{T}_n$ the operator that maps a vector $\sum_{k=0}^{\infty}{\mathcal G}_{k+1}z^k$ from $H_+$ to a finite dimensional vector $\sum_{k=1}^{n}\widetilde{\mathcal G}_{-k+1}z^k\in H_-$. These operators can be written as the superpositions $\widetilde T_n= C^{(n)}_{1,2}\circ C_{1,1}^{-1}\colon H_+\to H_-$, where $C^{(n)}_{1,2}$ is equal to the $n$-th cut of the block $C_{1,2}$ in~\eqref{bloc} of the first lower $n$-rows and with vanishing others. The operators $\widetilde{T}_n\colon H_+\to H_-$ are of finite rank, and therefore, compact. Their graphs ${W}_{\widetilde{T}_n}=(\id+\widetilde T_n)(H_+)\in\Gr_{\infty}(H)$ belong to the connected component of virtual dimension $0$.

\medskip

\noindent {\sc Step 2.} Observe that up to now there is no clear relation of operators $\widetilde{T}_n$, or their graphs with the Kirillov vector fields $L_k$ and $L_{-k}$. However, it is not hard to see, that  the quantities $\bar{\mathcal G}_k$, considered as functions of $\bar{\psi}$ are mapped to~$L_k[f]$ under the isomorphism $\iota$ for $k>0$. In Step 2  we are aimed to modifying $\widetilde{\mathcal G}_{-k}$, defined in~\eqref{gtilde} to ${G}_{-k}$ in such a way that  the isomorphism $\iota$ maps $\bar{G}_{-k}$ to the `non-positive' Kirillov vector fields $L_{-k}$. We will construct only $\bar{G}_0,\bar{G}_{-1},\bar{G}_{-2}$, and then, we extend the isomorphism $\iota$ to the Lie algebra isomorphism by defining $\bar{G}_{-(n+m)}(m-n)=\{\bar{G}_{-n},\bar{G}_{-m}\}$, $n,m\geq 0$.

Let us remind that the first 3 Virasoro generators written in affine coordinates are
\begin{itemize}
  \item $L_0[f](z)=\sum_{n=1}^{\infty}nc_n\partial_n$;
 \item $L_{-1}[f](z)=\sum_{n=1}^{\infty}\Big((n+2)c_{n+1}-2c_1c_n\Big)\partial_n$;
 \item $L_{-2}[f](z)=\sum_{n=1}^{\infty}\Big((n+3)c_{n+2}+(c_1^2-4c_2)c_{n}-a_{(n+2)}\Big)\partial_n,$
where the coefficient $a_n$ can be found from the recurrent relations
\begin{equation}\label{arec}
a_n=-\sum_{k=1}^{n}c_ka_{n-k},\quad a_0=1.
\end{equation}
 \end{itemize}

In order to construct $\bar G_k=\iota^{-1}(L_k)$, $k=0,-1$ we consider the coefficients $\widetilde{\mathcal G}$ from~\eqref{gtilde} as functions of $\bar\psi_{>0}$, and write $\bar\psi_0^*=\sum_{k=1}^{\infty} c_k\bar\psi_k$. We deduce that
$$
{G}_0=\widetilde{\mathcal G}_0-\psi^*_0,
\qquad
{G}_{-1}=\widetilde{\mathcal G}_{-1}-2 \bar c_1\psi^*_0.
$$
Since $\widetilde{\mathcal G}_{-2}=\sum_{k=1}^{\infty}(k+3) \bar c_{k+2}\psi_k$, we have
$$
{G}_{-2}=\widetilde{\mathcal G}_{-2}+\sum_{k=1}^{\infty}\Big((\bar c_1^2-4 \bar c_2)\bar c_{k}-\bar a_{(k+2)}\Big)\psi_k.
$$

Let us write this in terms of operators. Let
$$
\widetilde{B}^{(0)}=\left(\begin{array}{ccccc}
\cdots&\cdots&\cdots &\cdots&\cdots  \\
\cdots& 0 &0&0 & 0  \\
\cdots& 0 &0&0 & 0  \\
\cdots&-{c}_4& -{c}_3 & -{c}_2 & -{c}_1 
 \end{array}\right), 
 $$
$$
\widetilde{B}^{(1)}=\left(\begin{array}{ccccc}
\cdots&\cdots&\cdots &\cdots&\cdots  \\
\cdots& 0 &0&0 & 0  \\
\cdots& -2{c}_1{c}_4 &-2{c}_1{c}_3& -2{c}_1{c}_2 & -2{c}_1{c}_1  \\
 \cdots&-{c}_4& -{c}_3 & -{c}_2 & -{c}_1 
 \end{array}\right),
 $$
$$
\widetilde{B}^{(2)}=\left(\begin{array}{ccccc}
  \cdots&\cdots&\cdots &\cdots&\cdots  \\
\cdots& 0 &0&0 & 0  \\
 \cdots& ( c_1^2-4 c_2) c_{4}-\alpha_{6} &( c_1^2-4 c_2)c_{3}- \alpha_{5}&( c_1^2-4c_2)c_{2}-\alpha_{4}  & 2{c}_1 -6{c}_1{c}_2 +{c}_3 \\
\cdots& -2{c}_1{c}_4 &-2{c}_1{c}_3& -2{c}_1{c}_2 & -2{c}_1{c}_1  \\
\cdots&-{c}_4& -{c}_3 & -{c}_2 & -{c}_1 
 \end{array}\right),
$$
$$
 C_{1,2}^{(0)}=\left(\begin{array}{ccccc}
  \cdots&\cdots&\cdots &\cdots&\cdots\\
\cdots& 0 &0&0 & 0  \\
\cdots& 0 &0&0 & 0  \\
 \cdots&5{c}_4& 4{c}_3 & 3{c}_2 & 2{c}_1
 \end{array}\right),\quad C_{2,1}^{(1)}=\left(\begin{array}{ccccc}
 \cdots&\cdots&\cdots &\cdots&\cdots  \\
\cdots& 0 &0&0 & 0  \\
 \cdots&6{c}_5& 5{c}_4 & 4{c}_3 & 3{c}_2 \\
   \cdots&5{c}_4& 4{c}_3 & 3{c}_2 & 2{c}_1 
 \end{array}\right),
$$
$$
 C_{2,1}^{(2)}=\left(\begin{array}{ccccc}
 \cdots&\cdots&\cdots &\cdots&\cdots  \\
 \cdots& 0 &0&0 & 0  \\
 \cdots&7{c}_6& 6{c}_5 & 5{c}_4 & 4{c}_3 \\
\cdots&6{c}_5& 5{c}_4 & 4{c}_3 & 3{c}_2 \\
  \cdots&5{c}_4& 4{c}_3 & 3{c}_2 & 2{c}_1 
 \end{array}\right),
$$
where $a_n$ are given by~\eqref{arec}.
Then the operators $T_n$ such that their conjugates are $\bar T_n=(\widetilde B^{(n)}+C^{(n)}_{2,1})\circ C_{1,1}^{-1}$, are operators from $H_+$ to $H_-$ of finite rank and their graphs $W_{T_n}=(\id+T_n)(H_+)$ are elements of the component of virtual dimension $0$ in $\Gr_{\infty}(H)$.
We can choose a basis $\{e_0,e_1,e_2,\dots\}$ in  $W_{T_n}$ as a set of Laurent polynomials constructed by means of operators $T_n$ and $\bar C_{1,1}$ as
$$\{\psi_1,\psi_2,\dots\}\buildrel{\bar{C}_{1,1}}\over\longrightarrow \{G_1,G_2,\dots\}\buildrel{\id+T_n}\over\longrightarrow\{G_{-n+1},G_{-n+2},\dots, G_{0},G_1,G_2,\dots\},$$
projecting the canonical basis $\{1,0,0,\dots\}$, $\{0,1,0,\dots\}$, $\{0,0,1,\dots\}$,\dots:
\begin{eqnarray*}
e_0 &=& 1+\bar{c}_1\frac{1}{z}+(3\bar c_2-2\bar c_1^2)\frac{1}{z^2}+(5\bar c_3+2\bar c_1^3-6 \bar c_1\bar c_2)\frac{1}{z^3}+\dots +G_{-n+1}(\bar C_{1,1}(1,0,0,\dots))\frac{1}{z^n},\\
e_1& = &z+2\bar c_1+2\bar{c}_2\frac{1}{z}+(4\bar c_3-2\bar c_1\bar c_2)\frac{1}{z^2}+(6\bar c_4-5\bar c_2^2-2 \bar c_1\bar c_3+4\bar c_1^2 \bar c_2-\bar c_1^4)\frac{1}{z^3}+\dots\\ 
& & +\, G_{-n+1}(\bar C_{1,1}(0,1,0,\dots))\frac{1}{z^n},\\
e_2& = &z^2+2\bar c_1z+3\bar c_2+3\bar{c}_3\frac{1}{z}+(5\bar c_4-2\bar c_1\bar c_3)\frac{1}{z^2}+\\ 
& &+ (7\bar c_5-6\bar c_2\bar c_3+3 \bar c_1\bar c_2^2-2\bar c_1 \bar c_4+4\bar c_1^2\bar c_3-4\bar c_1^3c_2+\bar c_1^5)\frac{1}{z^3}+\dots\\ & &+\, G_{-n+1}(\bar C_{1,1}(0,0,1,\dots))\frac{1}{z^n},\\
\dots &\dots&\dots
\end{eqnarray*}
Let us  formulate the result as the following {\it main statement} of this section.

\begin{proposition}\label{graph}
The operator $(\id +T_n)$ defines a graph $W_{T_n}=\spn\{e_0,e_1,e_2,\dots\}$ in the Grassmannian $\Gr_{\infty}$ of virtual dimension 0. 
Given any $\psi=\sum_{k=0}^{\infty}\psi_{k+1}z^k\in H_+\subset H$, the function
\[
{G}(z)=\sum_{k=-n}^{\infty}{G}_{k+1}z^k=\sum_{k=0}^{\infty} \psi_{k+1}e_k,
\]
is an element of $W_{T_n}$. \end{proposition}

\begin{proposition}
In the autonomous case of the Cauchy problem \eqref{LK2}, when the function $p(z,t)$ does not depend on $t$, the pseudo-Hamiltonian $\mathcal H$ plays the role of time-dependent energy and $\mathcal H(t)=\bar{G}_0(t)+const$, where $\bar{G}_0\big|_{t=0}=0$. The constant is defined as
$\sum_{n=1}^{\infty}p_k\bar{\psi}_k(0)$.
\end{proposition}
\begin{proof}
In the autonomous case we have $\frac{d}{dt}\mathcal H=\frac{\partial}{\partial t}\mathcal H$. By straightforward calculation we assure that
$\frac{d}{dt}\bar{G}_0=\frac{\partial}{\partial t}\mathcal H$, which leads to the conclusion of the proposition. The constant is calculated by substituting $t=0$ in $\mathcal H$.
\end{proof}

\begin{remark} The Virasoro generator $L_0$ plays the role of the energy functional in CFT. In the view of the isomorphism $\iota$, the observable $\bar{G}_0=\iota^{-1}(L_0)$ plays an analogous role.
\end{remark}

As a {\it conclusion}, we constructed a countable family of curves $\Gamma_n\colon[0,T]\to\mathcal B$ in the trivial bundle $\mathcal B=T^{*}\mathcal F_0\otimes \mathbb C\,\times\,\Gr_{\infty}(H)$, such that the curve $\Gamma_n$ admits the form $\Gamma_n(t)=\Big(f(z,t),\psi(z,t),W_{T_n}(t)\Big)$, for $t\in[0,T]$ in the local trivialization. Here $\big(f(z,t),\bar\psi(z,t)\big)$ is the solution of the Hamiltonian system~(\ref{sys1}--\ref{sys2}). Each operator $T_n(t)\colon H_+\to H_-$ that maps ${\mathcal G}_{>0}$ to $\big({G}_0(t),{G}_{-1}(t),\ldots,{G}_{-n+1}(t)\big)$ defined for any $t\in[0,T]$, $n=1,2,\ldots$, is of finite rank and its graph $W_{T_n}(t)$  is a point in $\Gr_{\infty}(H)$ for any $t$. The graphs $W_{T_n}$  belong to the connected component of the virtual dimension $0$ for every time $t\in[0,T]$ and for fixed $n$.
Each coordinate $(G_{-n+1},\ldots, G_{-2},G_{-1},G_0,{G}_1,{G}_2,\ldots)$ of a point in the graph $W_{T_n}$ considered as a function of $\psi$ is isomorphic to the Kirilov vector fields $$(L_{-n+1},\ldots,L_{-2},L_{-1},L_0,L_1,L_1,L_2,\ldots)$$ under the isomorphism $\iota$.

 \section{$\tau$-function}

Remind that any function $g$ holomorphic in the unit disc, non vanishing on the boundary and normalized by $g(0)=1$ defines the multiplication operator $g\varphi$, $\varphi(z)=\sum_{k\in\mathbb Z}\varphi_kz^k$, that can be written in the matrix form
\begin{equation}\label{multip}
\left(
\begin{array}{cc}
a & b \\
0 & d
\end{array} \right) \left(
\begin{array}{c}
\varphi_{\geq 0}\\
 \varphi_{< 0}
\end{array} \right).
\end{equation}
All these upper triangular matrices form a subgroup $GL_{res}^{+}$ of the group of automorphisms $GL_{res}$ of the Grassmannian $\Gr_{\infty}(H)$.

With  any function $g$ and any graph $W_{T_n}$ constructed in the previous section (which is transverse to $H_-$) we can relate the $\tau$-function $\tau_{W_{T_n}}(g)$ by the following formula
$$
\tau_{W_{T_n}}(g)=\det(1+a^{-1}bT_n),$$
where $a,b$ are the blocks in the multiplication operator generated by $g^{-1}$. If we write the function $g$ in the form $g(z)=\exp(\sum_{n=1}^{\infty}t_nz^n)=1+\sum_{k=1}^{\infty}S_k(\tb)z^k$, where the coefficients $S_k(t)$ are the homogeneous elementary Schur polynomials, then the coefficients $t_n$ are called generalized times.
 For any fixed $W_{T_n}$ we get an orbit in $\Gr_{\infty}(H)$ of curves $\Gamma$ constructed in the previous section under the action of the elements of the subgroup  $GL_{res}^{+}$ defined by the function $g$. On the other hand,  the $\tau$-function defines a section in the determinant bundle over $\Gr_{\infty}(H)$ for any fixed $f\in\mathcal F_0$ at each point of the curve $\Gamma$.

\section{Baker-Akhiezer function, KP flows, and KP equation}

Let us consider the component $\Gr^{0}$ of the Grassmannian $\Gr_{\infty}$  of virtual dimension $0$, and let $g$ be a holomorphic function in $\mathbb D$ considered as an element of $GL_{res}^{+}$ analogously to  the previous section. Then $g$ is an upper triangular matrix with 1s on the principal diagonal. Observe that $g(0)=1$ and $g$ does not vanish on $S^1$. Given a point $W\in \Gr^{0}$ let us define
a subset $\Gamma^+ \subset GL_{res}^{+}$ as $\Gamma^+=\{g\in GL_{res}^{+}:\, g^{-1} W \text{\ is transverse to \ } H_-\}$. Then  there exists \cite{SegalWilson} a unique function $\Psi_W[g](z)$ defined on $S^1$, such that for each $g\in \Gamma^+$, the function $\Psi_W[g]$ is in $W$, and it admits the form
\[
\Psi_W[g](z)=g(z)\left(1+\sum_{k=1}^{\infty}\omega_k(g,W)\frac{1}{z^k}\right).
\]
The coefficients $\omega_k=\omega_k(g,W)$ depend both on $g\in \Gamma^{+}$ and on $W\in \Gr^{0}$, besides they are holomorphic on $\Gamma^+$ and extend to meromorphic functions on $GL_{res}^{+}$. The function $\Psi_W[g](z)$ is called the {\it Baker-Akhiezer function} of $W$. It plays a crucial role in the definition of the KP (Kadomtsev-Petviashvili) hierarchy which we will define later. We are going to construct the Baker-Akhiezer function explicitly in our case.

Let $W=W_{T_n}$ be a point of $\Gr^0$ defined in Proposition~\ref{graph}. Take a function $g(z)=1+a_1 z+a_2z^2+\dots\in \Gamma^+$, and
let us write the corresponding bi-infinite series for the Baker-Akhiezer function $\Psi_W[g](z)$ explicitly as 
\begin{eqnarray*}
\Psi_W[g](z)=\sum_{k\in\mathbb Z}\mathcal W_kz^k=(1 & + & a_1z+a_2z^2 + \dots)\left(1+\frac{\omega_1}{z}+\frac{\omega_2}{z^2}+\dots\right)=\\
= \dots &+& (a_2+a_3\omega_1+a_4\omega_2+a_5\omega_3+\dots)z^2\\
&+& (a_1+a_2\omega_1+a_3\omega_2+a_4\omega_3+\dots)z\\
&+& (1+a_1\omega_1+a_2\omega_2+a_3\omega_3+\dots)\\
&+& (\omega_1+a_1\omega_2+a_2\omega_3+\dots)\frac{1}{z}\\
&+& (\omega_2+a_1\omega_3+a_2\omega_4+\dots)\frac{1}{z^2}+\dots\\
\dots&+& (\omega_k+a_1\omega_{k+1}+a_2\omega_{k+2}+\dots)\frac{1}{z^k}+\dots\\
\end{eqnarray*}
The Baker-Akhiezer function for $g$ and $W_{T_n}$ must be of the form
\[
\Psi_{W_{T_n}}[g](z)=g(z)\left( 1+ \sum_{k=1}^{n}\omega_k(g)\frac{1}{z^k}\right)= \sum_{k=-n}^{\infty}\mathcal W_kz^k.
\]
For a fixed $n\in\mathbb N$ we truncate this bi-infinite series by putting $\omega_k=0$ for all $k>n$.
 In order to satisfy the definition of $W_{T_n}$, and determine the coefficients $\omega_1,\omega_2,\dots,\omega_n$, we must check 
 that there exists a vector $\{\psi_1,\psi_2,\dots\}$, such that
 $\Psi_{W_{T_n}}[g](z)=\sum_{k=0}^{\infty} e_k\psi_{k+1}$.
 First we express $\psi_k$ as linear functions  $\psi_k=\psi_k(\omega_1,\omega_2,\dots,\omega_n)$ by
 \begin{equation}\label{psi}
 (\psi_1,\psi_2,\psi_3,\dots)=\bar C_{1,1}^{-1}\Big(\mathcal W_0(\omega_1,\omega_2,\dots,\omega_n),\mathcal W_1(\omega_1,\omega_2,\dots,\omega_n),\dots\Big).
 \end{equation}
 Using Wronski formula we can write
 \begin{eqnarray*}
 \psi_1&=&\mathcal W_0-2\bar c_1\mathcal W_1-(3\bar c_2-4\bar c_1^2)\mathcal W_2-(4\bar c_3-12\bar c_2\bar c_1+8\bar c_1^3)\mathcal W_3+\dots,\\
\psi_2&=& \mathcal W_1-2\bar c_1\mathcal W_2-(3\bar c_2-4\bar c_1^2)\mathcal W_3-(4\bar c_3-12\bar c_2\bar c_1+8\bar c_1^3)\mathcal W_4+\dots,\\
\psi_3 &= &\mathcal W_2-2\bar c_1\mathcal W_3-(3\bar c_2-4\bar c_1^2)\mathcal W_4-(4\bar c_3-12\bar c_2\bar c_1+8\bar c_1^3)\mathcal W_5+\dots,\\
 \dots &\dots &\dots
 \end{eqnarray*}
 
 Next we define $\omega_1,\omega_2,\dots,\omega_n$ as functions of $g$ and $W_{T_n}$, or in other words, as functions
 of $a_k, \bar c_k$ by solving linear equations
 \begin{eqnarray*}
 \omega_1&=&\bar c_1\psi_1+2\bar c_2\psi_2+\dots k\bar c_k\psi_k+\dots,\\
 \omega_2&=& \sum_{k=1}^{\infty}\Big((k+2)\bar c_{k+1}-2\bar c_1\bar c_k\Big)\psi_k,\\
 \dots &\dots &\dots
 \end{eqnarray*}
 where $\psi_k$ are taken from \eqref{psi}. The solution exists and unique because of the general fact of the existence of the Baker-Akhiezer function. It is quite difficult task in general, however, in the case $n=1$, it is possible to write the solution explicitly in the matrix form. If
 \[
 A=\left(
\begin{array}{c}
\dots \\
3\bar c_3\\
2\bar c_2\\
\bar c_1
\end{array} \right)^T\bar C_{1,1}^{-1} \left(\begin{array}{c}
\dots \\
 a_3\\
a_2\\
a_1
\end{array} \right),\quad
B=\left(
\begin{array}{c}
\dots \\
3\bar c_3\\
2\bar c_2\\
\bar c_1
\end{array} \right)^T\bar C_{1,1}^{-1} \left(\begin{array}{c}
\dots \\
 a_2\\
a_1\\
1
\end{array} \right).
 \]
 then $\omega_1=\frac{B}{1-A}$.

 In order to apply further theory of integrable systems we need 
to change variables $a_n\to a_n(\tb)$, $n>0$, $\tb=\{t_1,t_2,\dots\}$ in the following way
\[
a_n=a_n(t_1,\dots, t_n)=
S_n(t_1,\dots, t_n),
\]
where $S_n$ is the $n$-th elementary homogeneous Schur polynomial
\[
1+\sum\limits_{k=1}^{\infty}S_k(\tb)z^k=\exp\left(\sum\limits_{k=1}^{\infty}t_kz^k\right)=e^{\xi(\tb,z)}.
\]
In particular,
\[
S_1=t_1,\quad S_2=\frac{t_1^2}{2}+t_2,\quad S_3=\frac{t_1^3}{6}+t_1t_2+t_3,
\]
\[
S_4=\frac{t_1^4}{24}+\frac{t_2^2}{2}+\frac{t_1^2t_2}{2}+t_1t_3+t_4.
\]
Then the Baker-Akhiezer function corresponding to the graph $W_{T_n}$ is written as
\[
\Psi_{W_{T_n}}[g](z)=\sum_{k=-n}^{\infty}\mathcal W_kz^k=e^{\xi(\tb, z)}\left(1+\sum_{k=1}^{n}\frac{\omega_k(\tb,W_{T_n})}{z^k}\right),
\]
and $\tb=\{t_1,t_2,\dots\}$ is called the vector of generalized times.
It is easy to see that 
\[
\partial_{t_k}a_m=0,\quad\text{for all \ }m=1,2\dots,k-1,
\]
$\partial_{t_k}a_m=1$ and
\[
\partial_{t_k}a_m=a_{m-k},\quad\text{for all \ }m>k.
\]
In particular, $B=\partial_{t_1}A$. Let us denote $\partial:=\partial_{t_1}$. Then in the case $n=1$ we have
\begin{equation}\label{omega}
\omega_1=\frac{\partial A}{1-A}.
\end{equation}

Now we consider the associative algebra of pseudo-differential operators $\mathcal A=\sum_{k=-\infty}^na_k\partial^k$ over the space of smooth functions, where the derivation symbol $\partial$ satisfies the Leibniz rule and the integration symbol and its powers satisfy
the algebraic rules $\partial^{-1}\partial=\partial \partial^-1=1$ and $\partial^{-1}a$ is the operator $\partial^{-1}a=\sum_{k=0}^{\infty}(-1)^k(\partial^k a)\partial^{-k-1}$ (see; e.g., \cite{Dickey}). 
The action of the operator $\partial^{m}$, $m\in\mathbb Z$, is well-defined over the function $e^{\xi(\tb,z)}$, where $\xi(\tb,z)=\sum_{k=1}^{\infty}t_kz^k$,  so that  the function $e^{\xi(\tb,z)}$ is the eigenfunction of the operator $\partial^m$ for any integer $m$, i.e., it satisfies the equation
\begin{equation}\label{eigen}
\partial^me^{\xi(\tb,z)}=z^m e^{\xi(\tb,z)},\quad m\in\mathbb Z,
\end{equation}
see; e.g., \cite{Babelon, Dickey}. 
As usual, we identify  $\partial=\partial_{t_1}$, and $\partial^0=1$.

 Let us introduce the dressing operator $\Lambda=\phi\partial \phi^{-1}=\partial+\sum_{k=1}^{\infty}\lambda_k\partial^{-k}$, where $\phi$ is a pseudo-differential operator
$\phi=1+\sum_{k=1}^{\infty}w_k(\tb)\partial^{-k}$. 
The operator $\Lambda$ is defined up to the multiplication on the right by a series $1+\sum_{k=1}^{\infty}b_k\partial^{-k}$ with constant coefficients $b_k$. 
The $m$-th KP flow is defined by making use of the vector field
\[
\partial_m\phi:=-\Lambda^m_{<0}\phi,\quad \partial_m=\frac{\partial}{\partial t_m},
\]
and the flows commute. In the Lax form the KP flows are written as 
\begin{equation}\label{Lax}
\partial_m\Lambda=[\Lambda^m_{\geq 0},\Lambda].
\end{equation}
If $m=1$, then $\partial \Lambda=[\partial, \Lambda]=\sum_{k=1}^{\infty}(\partial{\lambda_k})\partial^{-k}$, which justifies the identification
$\partial=\partial_{t_1}$.

Thus, the Baker-Akhiezer function $\Psi_{W_{T_n}}[g](z)$  admits the form $\Psi_{W_{T_n}}[g](z)=\phi \exp(\xi(\tb,z))$ where $\phi$ is a pseudo-differential operator
$\phi=1+\sum_{k=1}^{n}\omega_k(\tb,W_{T_n})\partial^{-k}$. 
The  function $\Psi_{W_{T_n}}[g](z)$ becomes the eigenfunction of the operator $\Lambda^m$, namely $\Lambda^m w=z^m w$, for $m\in\mathbb Z$. Besides, $\partial_m w= \Lambda^m_{>0}w$.  In the view of \eqref{eigen} we can write this function
as previously,
\[
\Psi_{W_{T_n}}[g](z)=(1+\sum_{k=1}^{n}\omega_k(\tb, W_{T_n})z^{-k})e^{\xi(\tb,z)}.
\]

\begin{proposition}
Let $n=1$, and let the Baker-Akhiezer function be of the form
\[
\Psi_{W_{T_n}}[g](z)=e^{\xi(\tb,z)}\left(1+\frac{\omega}{z}\right),
\]
where $\omega=\omega_1$ is given by the formula \eqref{omega}. Then 
\[
\partial \omega=\frac{\partial^2 A}{1-A}+\left(\frac{\partial A}{1-A}\right)^2
\]
is a solution to the KP equation with  the Lax operator $L=\partial^2-2(\partial \omega)$.
\end{proposition}
\begin{proof}
First of all, we observe that 
\[
\Lambda^2_{\geq 0}=\partial^2+2\lambda_1,\quad \Lambda^3_{\geq 0}=\partial^3+3\lambda_1\partial+3(\partial \lambda_1)+3\lambda_2.
\]
Given $\phi=1+\omega\partial^{-1}$, we are looking for the coefficient $\lambda_1$
checking the equality $\partial_{t_2}\Psi_{W_{T_n}}[g]=L\, \Psi_{W_{T_n}}[g]$,
for the Lax operator $L=\Lambda^2_{\geq 0}=\partial^2+2\lambda_1$.

First of all, we need some auxiliary calculations
\[
\partial_{t_k}A=\partial^kA, \quad k=1,2,\dots;\qquad \partial_{t_2}\omega=\partial^2\omega-2\omega\partial\omega;
\]
\[
\partial_{t_2}\Psi=z^2\Psi+g\frac{\partial_{t_2} \omega}{z};
\]
\[
\partial^2\Psi=z^2\Psi+\frac{g}{z}(2z\partial \omega+2\omega\partial \omega+\partial_{t_2}\omega).
\]
Then, comparing latter two equalities we conclude that $\partial_{t_2}\Psi=\partial^2\Psi-(2\partial\omega)\Psi$, and $\lambda_1=-\partial\omega$.
Now we use the formula for the KP hierarchy \eqref{Lax} and write the time evolutions 
\begin{eqnarray*}
\partial_{t_2}\lambda_1&=&\partial^2\lambda_1+2\partial\lambda_2,\\
\partial_{t_2}\lambda_2&=&\partial^2\lambda_2+2\partial\lambda_3+2\lambda_1\partial\lambda_1,\\
\partial_{t_3}\lambda_1&=&\partial^3\lambda_1+3\partial^2\lambda_2+3\partial\lambda_3+6\lambda_1\partial\lambda_1.
\end{eqnarray*}
Finally, eliminating $\lambda_2$ and $\lambda_3$ we arrive at the first equation (KP equation) in the KP hierarchy for $\partial\omega$
\[
3\partial^2_{t_2}\lambda_1=\partial (4\partial_{t_3}\lambda_1-12\lambda_1\partial\lambda_1-\partial^3\lambda_1).
\]
The latter is a standard procedure, see; e.g., \cite{Babelon}.
\end{proof}

Of course, one can express the Baker-Akhiezer function  directly from the $\tau$-function by the Sato formula
\[
\Psi_{W_{T_n}}[g](z)=e^{\xi(\tb,z)}\frac{\tau_{W_{T_n}}(t_1-\frac{1}{z},t_2-\frac{1}{2z^2},t_3-\frac{1}{3z^3},\dots)}{\tau_{W_{T_n}}(t_1,t_2,t_3\dots)},
\]
or applying the {\it vertex operator} $V$ acting on the Fock space $\mathbb C[\tb]$ of homogeneous polynomials
\[
\Psi_{W_{T_n}}[g](z)=\frac{1}{\tau_{W_{T_n}}}V\tau_{W_{T_n}},
\]
where
\[
V=\exp\left(\sum\limits_{k=1}^{\infty}t_kz^k\right)\exp\left(-\sum\limits_{k=1}^{\infty}\frac{1}{k}\frac{\partial}{\partial t_k}z^{-k}\right).
\]
In the latter expression $\exp$ denotes the formal exponential series and $z$ is another formal variable that commutes with all Heisenberg operators $t_k$ and $\frac{\partial}{\partial t_k}$. Observe that
the exponents in $V$ do not commute and the product of exponentials is calculated by the Baker-Campbell-Hausdorff formula.
The operator $V$ is a vertex operator in which the coefficient $V_k$ in the expansion of $V$ is a well-defined linear operator on the space  $\mathbb C[\tb]$. The Lie algebra of operators spanned by $1,t_k,\frac{\partial}{\partial t_k}$, and $V_k$, is isomorphic to the affine Lie algebra $\hat{\mathfrak{s}\mathfrak{l}}(2)$.
The vertex operator  $V$ plays a central role in the highest weight representation of affine Kac-Moody algebras \cite{Kac, Moody}, and can be interpreted as the infinitesimal B\"acklund transformation for the Korteweg--de Vries equation~\cite{Date}.

The vertex operator $V$ recovers the Virasoro algebra in the following sense. Taken in two close points $z+\lambda/2$ and $z-\lambda/2$ the operator product
can be expanded in to the following formal Laurent-Fourier series
\[
:V(z+\frac{\lambda}{2})V(z-\frac{\lambda}{2}):=\sum\limits_{k\in \mathbb Z}W_k(z)\lambda^{k},
\]
where $:a b:$ stands for the bosonic normal ordering.
Then $W_2(z)=T(z)$ is the stress-energy tensor which we expand again as
\[
T(z)=\sum\limits_{n\in \mathbb Z}L_n(\tb)z^{n-2},
\]
where the operators $L_n$ are the Virasoro generators in the highest weight representation over  $\mathbb C[\tb]$. Observe that the generators $L_n$
span the full Virasoro algebra with central extension and with the central charge~1. This can be also interpreted as a quantization of the shape evolution.


\begin{thebibliography}{99}

\bibitem{AM}
H.~Airault, P.~Malliavin, {\it Unitarizing probability measures for representations of Virasoro algebra.} J. Math. Pures Appl. {\bf 80} (2001), no. 6, 627--667.

\bibitem{AiraultNeretin}
H.~Airault, Yu.~Neretin, {\it On the action of Virasoro algebra on the space of univalent functions.}
 Bull. Sci. Math. {\bf 132} (2008), no. 1, 27--39.

\bibitem{Babelon}
O.~Babelon, D.~Bernard, M.~Talon, {\it Introduction to classical integrable systems}, Cambridge Univ. Press, 2003.

\bibitem{BB}
M.~Bauer, D.~Bernard, {\it Conformal field theories of stochastic Loewner evolutions},
Comm. Math. Phys. {\bf 239} (2003), no. 3, 493--521.

 \bibitem{Bieb}
L.~Bieberbach, {\it \"Uber die Koeffizienten derjenigen
Potenzreihen, welche eine schlichte Abbildung des Einheitskreises
vermitteln}, S.-B. Preuss. Akad. Wiss. (1916), 940--955.

\bibitem{Boggess}
A.~Boggess, {\it CR manifolds and the tangential Cauchy-Riemann complex.} Studies in Advanced Mathematics. CRC Press, Boca Raton, FL, 1991. 364 pp.

 \bibitem{Bott}
 R.~Bott, {\it On the characteristic classes of groups of diffeomorphisms}. Enseignment Math. (2) {\bf 23} (1977), no. 3-4, 209--220.

\bibitem{Contreras}
F.~Bracci, M.~D.~Contreras, S.~D{\'\i}az-Madrigal, {\it Evolution Families and the Loewner Equation I: the unit disc}, Journal f\"ur die reine und angewandte Mathematik (to appear); arXiv: 0807.1594.

 \bibitem{Branges}
 L.~de~Branges, {\it A proof of the Bieberbach conjecture}, Acta
Math. {\bf 154} (1985), no. 1-2, 137--152.

\bibitem{Date}
E.~Date,   M.~Kashiwara,   T.~Miwa,  {\it Vertex operators and functions: transformation groups for soliton equations II},  Proc. Japan Acad. Ser. A Math. Sci. , {\bf 57}  (1981),  387--392

\bibitem{DJKM83}
E.~Date, M.~Kashiwara, M.~Jimbo, T.~Miwa, {\it Transformation groups for soliton equations.}
Nonlinear integrable systems-classical theory and quantum theory (Kyoto, 1981), 39-119, World Sci. Publishing, Singapore, 1983.

\bibitem{Dickey}
L.~A.~Dickey, {\it Soliton equations and Hamiltonian systems} (2-nd ed.), Adv. Ser. Math. Phys, Vol. 26, World Scientific, 2003.

\bibitem{Doug}
Douglas~R.~G. {\it Banach algebra techniques in operator theory.} Second edition. Graduate Texts in Mathematics, {\bf 179.} Springer-Verlag, New York, 1998. 194 pp.

\bibitem{Faddeev}
L.~D.~Faddeev, {\it Discretized Virasoro algebra},   Contemp. Math., 391, Amer. Math. Soc., Providence, RI, 2005, 59--67.

\bibitem{Friedrich}
R.~Friedrich, {\it The global geometry of stochastic Loewner evolutions}, Probabilistic approach to geometry,  Adv. Stud. Pure Math., 57, Math. Soc. Japan, Tokyo, 2010, 79--117.

\bibitem{FriedrichWerner}
R.~Friedrich, W.~Werner, {\it Conformal restriction, highest-weight representations
and SLE}, Comm. Math. Phys. {\bf 243} (2003), no. 1, 105--122.

\bibitem{GelfandFuchs}
I.~M.~Gel'fand, D.~B.~Fuchs, {\it Cohomology of the Lie algebra of vector
fields on the circle},   Functional Anal. Appl. {\bf 2} (1968), no.4, 342--343.

\bibitem{Gervais}
J.-L.~Gervais, {\it  Infinite family of polynomial functions of the Virasoro generators with vanishing Poisson brackets}, Phys. Lett. B {\bf 160} (1985), no. 4-5, 277--278.

\bibitem{Goryainov}
V.~V.~Goryainov, {\it Fractional iterates of functions that are analytic in the unit disk with given fixed points}, Mat.
Sb. {\bf 182 (9)} (1991) 1281--1299; Engl. Transl. in Math. USSR-Sb. {\bf 74 (1)} (1993) 29--46.

\bibitem{GustVas}
B.~Gustafsson, A.~Vasil'ev, {Conformal and potential analysis in Hele-Shaw cells}, Birkh\"auser, 2006.

\bibitem{Kac}
V.~G.~Kac,   {\it Simple irreducible graded Lie algebras of finite growth}  Math. USSR Izv., {\bf 2}  (1968), 1271--1311;  Izv. Akad. Nauk USSR Ser. Mat., {\bf 32}  (1968), 1923--1967.

\bibitem{Unirreps}
A.~A.~Kirillov, {\it Geometric approach to discrete series of unirreps for vir}, J. Math. Pures Appl. {\bf 77} (1998) 735--746

\bibitem{KYu}
A.~A.~Kirillov, D.~V.~Yuriev, {\it Representations of the Virasoro algebra by the orbit method}, J. Geom. Phys. {\bf 5 (3)}
(1988) 351--363.

\bibitem{Kufarev}
P.~P.~Kufarev, {\it On one-parameter families of analytic
functions}, Rec. Math. [Mat. Sbornik] N.S. {\bf 13(55)} (1943),
87--118.

\bibitem{Lebrun}
C.~Lebrun, {\it A K\"ahler structure on the space of string world sheets}, Class. Quant. Gravity {\bf 10} (1993), 141--148.

\bibitem{Lempert1}
L.~Lempert,  {\it Loop spaces as complex manifolds}, J. Differential Geom. {\bf 38} (1993), 519--543.

\bibitem{Lempert}
L.~Lempert,  {\it The Virasoro group as a complex manifold},
Math. Res. Lett. {\bf 2} (1995), 479--495.

\bibitem{Loewner}
K.~L\"owner, {\it Untersuchungen \"uber schlichte konforme
Abbildungen des Einheitskreises}, Math. Ann. {\bf 89} (1923),
103--121.

\bibitem{MPV}
I.~Markina, D.~Prokhorov, A.~Vasil'ev, {\it Sub-Riemannian geometry of the coefficients of univalent functions},  J. Funct. Analysis {\bf 245} (2007),  no. 2, 475--492.

\bibitem{MarkVas2010}
I.~Markina, A.~Vasil'ev, {\it Virasoro algebra and dynamics in the space of univalent functions}. Contemporary Math. {\bf 525} (2010), 85--116.

\bibitem{Milnor}
J.~Milnor, {\it Remarks on infinite-dimensional Lie groups}, pp. 1007--1057
in: `Relativit\'e, Groupes et Topologie II', B. DeWitt and R. Stora (Eds),
North-Holland, Amsterdam, 1984.

\bibitem{Mineev}
M.~Mineev-Weinstein, P.~B.~Wiegmann,  A.~Zabrodin, {\it Integrable structure of interface dynamics},  Phys. Rev. Letters  {\bf 84}  (2000), no. 22, 5106--5109.

\bibitem{Moody}
R~V.~Moody,  {\it A new class of Lie algebras},  J. Algebra {\bf 10}  (1968), 211--230.

\bibitem{Mumford}
D.~Mumford, {\it Pattern theory: the mathematics of perception}, Proceedings ICM 2002, vol. 1, 401-422.

\bibitem{MPPR}
J.~Mu\~{n}oz Porras, F.~Pablos Romo, {\it Generalized reciprocity laws.}
Trans. Amer. Math. Soc. 360 (2008), no. {\bf 7}, 3473--3492.

\bibitem{NN}
A.~Newlander, L.~Nirenberg, {\it Complex analytic coordinates in almost complex manifolds}, Ann. Math. {\bf 65} (1957), 391--404.


\bibitem{Polchinski}
J.~Polchinski, {\it String theory}, Cambridge Univ. Press, 1998.

\bibitem{Pommerenke1}
Ch.~Pommerenke, {\it \"Uber die Subordination analytischer
Funktionen}, J. Reine Angew. Math. {\bf 218} (1965), 159--173.


\bibitem{Pommerenke2}
Ch.~Pommerenke, {\it Univalent functions, with a chapter on
quadratic differentials by G.~Jensen}, Vandenhoeck \& Ruprecht,
G\"ottingen, 1975.

\bibitem{ProkhVas}
D.~Prokhorov, A.~Vasil'ev, {\it Univalent functions and integrable systems},  Comm. Math. Phys. {\bf 262} (2006), no. 2, 393--410.


\bibitem{PrSe}
Pressley~A.; Segal~G. {\it Loop groups.} Oxford Mathematical Monographs. Oxford Science Publications. The Clarendon Press, Oxford University Press, New York, 1986. 318 pp.

\bibitem{Raboin}
P.~Raboin, {\it  Le probl{\`e}me du $\bar\partial$ sur en espace de Hilbert}, Bull. Soc. Math. France {\bf 107} (1979), 225--240.

\bibitem{RS}
Reed~M.; Simon~B. {\it Methods of modern mathematical physics. I. Functional analysis.} Second edition. Academic Press, Inc. [Harcourt Brace Jovanovich, Publishers], New York, 1980. 400 pp.

\bibitem{Sato}
M.~Sato, Y.~Sato, {\it Soliton equations as dynamical systems on infinite-dimensional Grassmann manifold}, Nonlinear Partial Differential Equations in Applied Science Tokyo, 1982, North-Holland Math. Stud. vol. 81, North-Holland, Amsterdam (1983), pp. 259--271.

\bibitem{SegalWilson}
G.~Segal, G.~Wilson, {\it Loop groups and equations of KdV type}, Publ. Math. IHES No. 61, 5 (1985).

\bibitem{SS}
 A.~C.~Schaeffer, D.~C.~Spencer, {\it Coefficient Regions for
Schlicht Functions (With a Chapter on the Region of the Derivative
of a Schlicht Function by Arthur Grad)}, American Mathematical
Society Colloquium Publications, Vol. 35. American Mathematical
Society, New York, 1950.

\bibitem{Volpato}
M.~Volpato, {\it Sulla derivabilit{\`a}, rispetto a valori iniziali ed a parametri, delle soluzioni dei sistemi di equazioni differenziali ordinarie del primo ordine}, 
Rend. Sem. Mat. Univ. Padova {\bf 28} (1958), 71--106. 

\bibitem{Witten}
E.~Witten, {\it Quantum field theory, Grassmannians, and algebraic curves.}
Comm. Math. Phys. {\bf 113} (1988), no. {\bf 4}, 529--600.

\end{thebibliography}
\end{document}